\documentclass[journal,twocolumn]{IEEEtran}
\usepackage[utf8]{inputenc}
\usepackage[english]{babel}
\usepackage{cite}
\usepackage{amsmath,amssymb,amsfonts}
\usepackage{algorithmic}
\usepackage{cases}
\usepackage{bbm}
\usepackage{textcomp}
\usepackage{lipsum}
\usepackage{float}\usepackage{fixmath}
\usepackage{mathtools}
\usepackage{mathrsfs}
\usepackage{amssymb}
\usepackage{booktabs}
\usepackage{multirow}
\usepackage{subfig}
\usepackage{graphicx}
\usepackage{epstopdf}
\usepackage{breqn}
\usepackage{xcolor}
\usepackage{enumerate}
\usepackage{breqn}
\usepackage[normalem]{ulem}
\usepackage{arydshln}

\usepackage{enumitem}
\usepackage{amsmath}

\usepackage{amsmath}
\usepackage{cite}
\usepackage{refcount}
\usepackage{flushend}
 
\usepackage{amsthm}
\newtheorem{theorem}{Theorem}[section]

\newtheorem{lemma}[theorem]{Lemma}
\newtheorem{remark}{Remark}
\definecolor{newblue}{rgb}{0,0.4,0.7}
\definecolor{DarkBlue}{rgb}{0.1,0.22,0.45}
\definecolor{DarkGreen}{rgb}{0.2,0.6,0.2}
\definecolor{DarkRed}{rgb}{0.8,0.3,0.3}

\usepackage[hyperfootnotes=false,colorlinks = true, linkcolor = black,citecolor = black]{hyperref}
\usepackage[capitalize]{cleveref}
\setlist{nolistsep}

\newtheorem{deff}{Definition}

\newtheorem{assumption}{Assumptions}

\usepackage{lettrine}

\begin{document}

\title{\LARGE Cubature Kalman Filter as a Robust State Estimator Against Model Uncertainty and Cyber Attacks in Power Systems }

\author{Tohid Kargar Tasooji, Sakineh Khodadadi

\thanks{ Tohid Kargar Tasooji is with the Department of Aerospace Engineering,
Toronto Metropolitan University, Toronto, ON M5B 2K3, Canada (e-mail:
tohid.kargartasooji@torontomu.ca). Sakineh Khodadadi is with the Department of Electrical and Computer
Engineering, University of Alberta, Edmonton,AB,  AB T6G 1H9, Canada (email: sakineh@ualberta.ca).}
}

\maketitle
\begin{abstract}
 It is known that the conventional estimators such as extended Kalman filter (EKF) and unscented Kalman filter (UKF) may provide favorable performance; However, they may not guarantee the robustness against model uncertainty and cyber attacks.
In this paper, we compare the performance of cubature Kalman filter (CKF) to the conventional nonlinear estimator, the EKF, under the affect of model uncertainty and cyber-attack. We show that the CKF has better estimation accuracy than the EKF under some conditions. In order to verify our claim, we have tested the performance various nonlinear estimators on the single machine infinite-bus (SMIB) system under different scenarios. We show that (1) the CKF provides better estimation results than the EKF; (2) the CKF is able to detect different types of cyber attacks reliably which is superior to the EKF.
\end{abstract}
\IEEEoverridecommandlockouts

%
\IEEEpeerreviewmaketitle

\section{Introduction}
\lettrine{C}{\large yber-physical} systems (CPSs) are systems that integrate sensing, communication, control, computation and physical processes \cite{1, 50, 51, 52, 53, 54, 55, 56 , 57, 58, 59, 60, 61, 62, 63}. CPSs are now ubiquitous in multiple areas including chemical processes, smart grids, mine monitoring, intelligent transportation, precision agriculture, aerospace, etc. \cite{2,3,4}.  One problem affecting CPSs, inherent in the use of a communication network to establish communication between system components, is their susceptibility to cyber attacks, \cite{5,6,7}.
Today’s threat of cyber attacks to CPSs is increasing at an alarming pace, and there is a consequent urgency in protecting CPSs against cyber attacks.

Cyber Attacks to CPSs have received much attention from the research community. Teixeira et al. \cite{8} analyzed different attack models, including Denial of Service (DoS), replay, False data injection (FDI) and zero dynamic attacks.
Cardenas et al. \cite{9} investigated cyber attacks
to measurement and actuator data integrity and availability. The authors present two types of CPS attacks: DoS
and deception attacks. DoS attacks obstruct the communication
channels preventing the exchange of information between system components. Deception attacks affect the integrity of data by manipulating its content. Mo et al. \cite{10} consider a form of deception attack known as replay-attacks. The authors introduce a replay attack model on CPS and analyze the performance of the control system under attack.\\

In this article, our primary interest is robust state estimation of nonlinear systems under cyber attacks. Most studies in nonlinear state estimation focus on extended Kalman filter (EKF) \cite{11}, \cite{12}, unscented Kalman filter (UKF) \cite{13,14,15,16,17}, square-root unscented Kalman filter (SR-UKF) \cite{18,19,20,21},  extended particle filter \cite{22,23} and ensemble Kalman filter \cite{24}.  

Guo et al. \cite{25} propose a new linear attack strategy on remote state estimation and present the corresponding feasibility constraint to ensure that the attack can successfully bypass a $\chi^{2}$-detector. Shi et al. \cite{26} propose a stochastic modeling framework for CPSs considering the effect of adversary attacks. Based on this framework, the authors formulate and solve a joint state and attack estimation problem. Xie et al. \cite{27} model the economic impact of malicious data attacks on real-time market operations. They examined the economic impact of FDI attacks against state estimation. Yang et al. \cite{28} present an appropriate phasor measurement unit (PMU) placement for reliable state estimation against
FDI attacks. Gandhi et al. \cite{29} propose a new robust Kalman filter based on the generalized
maximum-likelihood-type estimate against outliers. Taha et al. \cite{30} present a risk mitigation strategy, based on Sliding-Mode Observer and an attack detection
filter, against threat levels from the grid's unknown inputs and cyber-attacks. Karimipour et al. \cite{31} presented a robust dynamic state estimation algorithm based on EKF against FDI attacks in power systems. Zhao et al. \cite{32} introduce a general constrained robust dynamic state estimation (DSE) framework based on UKF to address various equality and inequality constraints while is robust against measurement noise and bad data. Manandhar et al. \cite{33} design a framework for the linear system using
KF  together with the $\chi^{2}$-detector and Euclidean detector which is able to detect different types of faults and attacks. Qi et al. \cite{34} compare different types of Kalman filters and dynamic observers under model uncertainty and malicious cyber attacks on nonlinear systems. \\

One problem associated with state estimators and their application is modelling error. When the mathematical model deviates from the actual plant being monitored, model predictions deviate from the actual state and therefore compromise their use in any application. This problem is more significant in the case of nonlinear systems, given that all of the filters known in the literature rely on system approximations of some form. The possible presence of cyber attacks therefore make this problem more critical.  Based on these observations, in this paper we focus on robust state estimation under cyber attacks and compare the performance of two well known filters, namely the EKF and the CKF.

The focus of this article is to draw a comparative study of the most popular nonlinear filters and their tolerance to both modelling error and cyber attacks. More specifically, we compare the EKF to the CKF.
First, we study the effect of modelling error. Based on our analytical analysis, we argue that the CKF provides better estimates than the EKF under certain conditions.  
 We then consider the effects of cyber attacks and once again, we argue that the CKF outperforms the EKF when used to detect cyber attacks. 
We show that the CKF can reliably detect cyber attacks, including random attacks, DoS, replay  and FDI attacks, and outperforms the EKF in all type of attacks.
We illustrate our claims via extensive simulations of the two filters under various conditions, including parameter uncertainties, uncertainty in measurement and noises as well as cyber attacks, using a fourth order synchronous machine.\\

The rest of the paper is organized as follows: Section 
II presents dynamic state estimation for nonlinear systems based on EKF and CKF. Section III presents the effect of modeling error on each estimator. Section IV presents the modelling of different types of cyber-attacks and the affect of random attacks on each estimator. Section V describes two cyber-attack detectors. Section VI presents the performance analysis of each estimator under different scenarios. Finally, Section VII presents the conclusions. 

\section{Nonlinear State Estimators}
Throughout the rest of the paper we consider a nonlinear discrete-time system model of the form:
\begin{equation} 
\left\{
                \begin{array} {l} \label{eq:equation3}
                x_{k} = f(x_{k-1},u_{k-1})+w_{k-1}\\
                 {y_{k}} = h(x_{k},u_{k})+v_{k}
                \end{array}
\right.
\end{equation}
where $x_{k}$ is the system state vector, $u_{k}$ is the input vector, $w_{k}$ is the process (random state) noise, $y_{k}$ is the noisy
observation or measured output vector, and  $v_{k}$ is the
measurement noise. The noise sequences $w_{k}$ and $v_{k}$ are assumed to be white, Gaussian, and independent of each other:
\begin{align}
&\left\{
                \begin{array} {l} \label{eq:equation4}
                E[w_{k}] = 0\\
                E[w_{k}w_{k}^T] = Q_{k}\\
                \end{array}
\right.
\\
&\left\{
                \begin{array} {l} \label{eq:equation5}
                E[v_{k}] = 0\\
                E[v_{k}v_{k}^T] = R_{k}\\
                \end{array}
\right.
\\
&\left\{
                \begin{array} {l} \label{eq:equation6}
                E[w_{k}w{j}^T] = 0 \ \ \ \ \ (\texttt{for}\  j\ne k)\\
                E[v_{k}v_{k}^T] = 0  \ \ \ \ \ (\texttt{for}\  j\ne k)\\
                \end{array}
\right.
\\
\label{eq:equation7}
& ~~~~E[w_{k}v_{j}^T] = 0
\end{align}
Equations \eqref{eq:equation4} and \eqref{eq:equation5} show that $w_{k}$ and $v_{k}$ have a zero
mean, with covariance matrices $Q_{k}$ and $R_{k}$, respectively. Equation \eqref{eq:equation6} indicates that the values of $v_{k}$ (respectively, $w_{k}$), at different time instants are uncorrelated, while equation \eqref{eq:equation7}
implies that the process (state) and measurement (observation)
noises are not cross-correlated.  

We will consider the following filters:

\subsection{\textbf{EXTENDED KALMAN FILTER}}
EKF is a well established nonlinear filter that has been used in multiple application over the last four decades. Some of its limitations include the following \cite{37,38,39}:
\begin{enumerate}
\item The EKF algorithm is based on a first-order Taylor expansion and therefore works well with relatively mild nonlinear models but performance can degrade when dealing with highly nonlinear systems.
\item The computation of the Jacobian matrices used in the linearization is often difficult and error-prone.
\item It suffers from the so-called \emph{curse of dimensionality} or
divergence or both; {\it i.e.} state estimates deviate from the true state for high order nonlinear models. 
\end{enumerate}

The discrete-time EKF algorithm for the state estimation is summarized in the appendix \cite{40}.

\subsection{\textbf{CUBATURE KALMAN FILTER}}
 CKF algorithm was proposed by Arasaratnam and Haykin \cite{38}, to resolve the limitations encountered in the EKF. The CKF algorithm consists of a numerical approximation of the integral consistent of the product of a nonlinear function and a Gaussian function,  using a spherical-radial cubature rule. Unlike the EKF algorithm, the CKF captures the higher-order terms of the nonlinear system resulting in better estimates. A limitation encountered in the CKF algorithm is however, the possible loss of positive definiteness when computing the predictive covariance in the filter algorithm, which can halt the CKFalgorithm. 
The discrete-time CKF algorithm for state estimation is summarized in the appendix \cite{38}.

\subsection{\textbf{SQUARE-ROOT CUBATURE KALMAN FILTER}}
SCKF is square-root extension of CKF which propagates square-root factors of the predictive
and posterior error covariances. Some benefits are reported for SCKF:
\begin{enumerate}
\item Symmetry and positive (semi) definiteness
of the covariance are preserved.
\item Doubled-order precision.
\end{enumerate}
The discrete-time SCKF algorithm for state estimation is summarized in the appendix \cite{38}.

\section{Model Uncertainty} 
One of the significant challenges in dynamic state estimation (DSE) is the inaccuracy of the model and parameters which can consequently deteriorate state estimation \cite{41}. The nonlinear stochastic system \eqref{eq:equation3} under model uncertainty can be rewritten as follows:
\begin{equation} 
\left\{
                \begin{array} {l} \label{eq:equation8}

                x_{k} = f(x_{k-1},u_{k-1})+\Delta A_{k-1}x_{k-1}+w_{k-1}\\
                 {y_{k}} = h(x_{k},u_{k})+\Delta C_{k}x_{k}+v_{k}
                \end{array}
\right.
\end{equation}
Where $\Delta A_{k-1}x_{k-1}$ and $\Delta C_{k}x_{k}$ represent parameter uncertainties, 
$\Delta A_{k-1}$ and $\Delta C_{k}$ are unknown matrices satisfying $\left\|\Delta A_{k-1}\right\|\le \alpha$ and $\left\|\Delta C_{k}\right\|\le \beta$, respectively.

\begin{deff}
The estimate $\hat{x}_{k|k}$ of the system state $x_{k}$ is said to be stable in the minimum mean square error (MMSE) sense, if $\mathbb{E}\left\{\left\|x_{k}-\hat{x}_{k|k}\right\|^2\right\} $ is bounded, {\it i.e.}, $\mathrm{sup}_{k} \mathbb{E}\left\{\left\|x_{k}-\hat{x}_{k|k}\right\|^2\right\}  \le \infty $ 
\end{deff}

\begin{assumption}\label{Ass:1}
There exist positive real numbers $\widetilde{a}$, $\widetilde{c}$, $\underline{p}$, $\bar{p}$, $\widetilde{w}$, $\widetilde{v}$ $>0$ such that
the following bounds on various matrices are satisfied
for every $k\geq0$: (i) $\left\|A_{k-1}\right\|\le \widetilde{a}$, (ii) $\left\|C_{k}\right\|\le \widetilde{c}$, (iii) $\underline{p}I \le P_{k|k}\le \bar{p}I$, (iv) $\left\|\mathbb E \left\{w_{k-1}^Tw_{k-1} \right\}\right\|\le \widetilde{w}$, (v) $\left\|\mathbb E \left\{v_{k}^Tv_{k} \right\} \right\|\le \widetilde{v}$, (viii) $\left\|\Delta A_{k-1}\right\|\le \frac{\delta_{a}}{\left\|x_{k-1}\right\|}$, (ix) $\left\|\Delta C_{k}\right\|\le \frac{\delta_{c}}{\left\|x_{k}\right\|}$
\end{assumption}

\begin{assumption} \label{Ass:3}
The nonlinearity of function $f(x_{k-1},u_{k-1})$ satisfying by \eqref{eq:equation8} is mild such that, in Taylor series expansion around $a_{k-1}$
\begin{equation} 
\begin{array}{l}
 f(x_{k-1},u_{k-1})=f(a_{k-1},u_{k-1})+A_{k-1}(x_{k-1}-a_{k-1})\\+o(x_{k-1},a_{k-1},u_{k-1})
             \end{array}
\end{equation}
\begin{itemize}
\item  It has derivatives of all orders.
\item $f(x_{k-1},u_{k-1})$ is  real analytic on an open set $D$.
\item The higher order terms is bounded as a follows:
\begin{equation}
 \begin{array}{l} 
\left\|o(x_{k-1},a_{k-1},u_{k-1})\right\| \\ = \left\|\sum\limits_{n=2}^{\infty} \frac{f^n(a_{k-1})}{n!} (x_{k-1}-a_{k-1})^n  \right\| \le  \mu  
  \end{array} 
  \end{equation}
Where for a small positive value $\rho$, we have:
\begin{equation}
 \begin{array}{l} 
 \forall x_{k-1}\in D \ \  \  \ \left\|x_{k-1}-a_{k-1}\right\| \le \rho
  \end{array} 
  \end{equation}
  \end{itemize}
\end{assumption}

\begin{remark}
Note that $A_{k-1}$ is the state transition matrix which is the Jacobian matrix of the nonlinear function $f(.)$, $C_{k}$ is measurement matrix which is the Jacobian matrix of the nonlinear function $h(.)$ and $P_{k|k}$ is the covariance of the estimation error.
\end{remark}
\begin{assumption}\label{Ass:2}
There exist positive real numbers $\varepsilon_{\theta}$, $\varepsilon_{\phi}$ such that the nonlinear functions $\theta_{f}$ and $\phi_{h}$ are bounded via \label{item3}
 \begin{equation}
 \begin{array}{l} \label{eq:equation9}
\left\|\theta_{f}(e_{k-1|k-1})\right\|\le \varepsilon_{\theta}\left\|e_{k-1|k-1}\right\|^2
  \end{array}
  \end{equation}
 \begin{equation}
 \begin{array}{l} \label{eq:equation10}
\left\|\phi_{h}(e_{k-1|k-1})\right\|\le \varepsilon_{\phi}\left\|e_{k-1|k-1}\right\|^2
  \end{array}
  \end{equation}
\end{assumption}
\begin{lemma} \label{1}
 Consider the stochastic nonlinear system with parameter uncertainties given by \eqref{eq:equation8} and  let Assumption \ref{Ass:1} and Assumption \ref{Ass:2} be satisfied. Then the upper bound for the estimation error in the CKF, $e_{k|k}^c$, satisfies the following bound: 
 \begin{equation}
 \begin{array}{l} \label{eq:equation11}
\mathbb{E}\left\{\left\|e_{k|k}^{c}\right\|^2\right\} \le \frac{\bar{p}}{\underline{p}}(1-\frac{\gamma}{2})^k\mathbb{E}\left\{\left\|e_{0|0}\right\|^2\right\} + \frac{\xi^{c}}{\underline{p}\gamma}
  \end{array}
  \end{equation}
provided that the initial estimation error satisfies $\left\|e_{0|0}\right\| \le \epsilon $, and $\xi^{c}$ is sufficiently small. Note that $\xi^{c}$ is a function of noise covariance and the parameter uncertainty and $0 <\gamma \le 1$. \end{lemma}

\begin{proof}
Define a random variable $X_{i,k-1|k-1} \sim N(\hat{x}_{k-1|k-1},P_{k-1|k-1})$ as Cubature points with Gaussian distribution and also $\Delta x_{k-1|k-1} \sim N(0,P_{k-1|k-1})$ with the Gaussian distribution. With these assumptions we can compute cubature points as follows:
\begin{equation}
 \begin{array}{l}  \label{eq:equation12}
X_{i,k-1|k-1}=\hat{x}^c_{k-1|k-1}+\Delta x_{k-1|k-1}^i.
  \end{array}
  \end{equation}
Using equations \eqref{eq:equation8}, \eqref{eq:equation12}, \eqref{eq:equation68} and \eqref{eq:equation69} and considering the parameter uncertainty in the process model, the prediction error can be computed as a follows:
\begin{equation}
 \begin{array}{l} \label{eq:equation13}
e_{k|k-1}^{c}=x_{k}-\hat{x}^c_{k|k-1}=f(x_{k-1},u_{k-1}) + \Delta A_{k-1}x_{k-1}\\+w_{k-1}-\left\{ 
\frac{1}{m}\sum\limits_{m=1}^{m} f(\hat{x}^c_{k-1|k-1}+\Delta x_{k-1|k-1}^i,u_{k-1})\right\}
  \end{array}
  \end{equation}
Using Taylor series expansion around $\hat{x}_{k-1|k-1}^c$, we can write:
\begin{equation}
 \begin{array}{l} \label{eq:equation14}
e_{k|k-1}^c=A_{k-1}e_{k-1|k-1}^c+\theta_{f}(|e_{k-1|k-1}^c|)+w_{k-1}\\ -A_{k-1}\left\{ 
\frac{1}{m}\sum\limits_{m=1}^{m} \Delta x_{k-1|k-1}^i\right\}-\left\{ 
\frac{1}{m}\sum\limits_{m=1}^{m} \theta_{f}(|\Delta x_{k-1|k-1}^i|)\right\}\\ + \Delta A_{k-1}x_{k-1}= A_{k-1}e_{k-1|k-1}^c+ \theta_{f}(|e_{k-1|k-1}^c|)\\ +w_{k-1}+ \Delta A_{k-1}x_{k-1}.
  \end{array}
  \end{equation}
  
The estimation error under parameter uncertainties is given by:
  \begin{equation}
 \begin{array}{l} \label{eq:equation15}
e_{k|k}^c=x_{k}-\hat{x}^c_{k|k}=f(x_{k-1},u_{k-1})+\Delta A_{k-1}x_{k-1}\\ +w_{k-1}\ -\hat{x}^c_{k|k-1} -K_{k}(y_{k}-h(\hat{x}^c_{k|k-1}))= A_{k-1}e_{k-1|k-1} \\ +\theta_{f}(|e_{k-1|k-1}^c|)+ \Delta A_{k-1}x_{k-1} +w_{k-1}\\ -K_{k}(h(x_{k},u_{k}) -\left\{ 
\frac{1}{m}\sum\limits_{m=1}^{m} h(\hat{x}^c_{k-1|k-1}+\Delta x_{k-1|k-1}^i,u_{k-1})\right\}\\ +v_{k}+ \Delta C_{k}x_{k}),
  \end{array}
  \end{equation}
where $h(\hat{x}^c_{k|k-1})$ can be expanded in Taylor series around $\hat{x}_{k|k-1}$ as follows:
  \begin{equation}
 \begin{array}{l} \label{eq:equation16}
e_{k|k}^c=A_{k-1}e_{k-1|k-1}^c +\theta_{f}(|e_{k-1|k-1}^c|)+\Delta A_{k-1}x_{k-1}\\ +w_{k-1}-K_{k}(C_{k}e_{k|k-1}^c+\phi_{h}(e_{k-1|k-1}^c)+v_{k}+ \Delta C_{k}x_{k}\\ -C_{k}\left\{
\frac{1}{m}\sum\limits_{m=1}^{m} \Delta x_{k-1|k-1}^i\right\}-\left\{ 
\frac{1}{m}\sum\limits_{m=1}^{m} \phi_{h}(|\Delta x_{k-1|k-1}^i|)\right\}).
  \end{array}
  \end{equation}
  
Thus
  \begin{equation}
 \begin{array}{l} \label{eq:equation17}
e_{k|k}^c=\widetilde{A}e_{k-1|k-1}^c+L_{k}+N_{k},
  \end{array}
  \end{equation}
where the estimation error consist of three terms:
  \begin{equation}
 \begin{array}{l} \label{eq:equation18}
 \widetilde{A}=[I-K_{k}C_{k}]A_{k-1}\\
 L_{k}=[I-K_{k}C_{k}]\theta_{f}(e_{k-1|k-1}^c)-K_{k}\phi_{h}(e_{k-1|k-1}^c)\\ +[I-K_{k}C_{k}]\Delta A_{k-1}x_{k-1}- K_{k}\Delta C_{k}x_{k}\\
 N_{k}=[I-K_{k}C_{k}]w_{k-1}-K_{k}v_{k}.
  \end{array}
  \end{equation}
Define the Lyapunov function $V:\mathbb{R}^n\rightarrow \mathbb{R}$:
  \begin{equation}
 \begin{array}{l} \label{eq:equation19}
V(e_{k-1|k-1}^c)=e_{k-1|k-1}^T P_{k-1|k-1}^{-1} e_{k-1|k-1}.
  \end{array}
  \end{equation}
From Assumption \ref{Ass:1}.(iii) we have:
  \begin{equation}
 \begin{array}{l} \label{eq:equation20}
\frac{1}{\bar{p}}\left\|e_{k-1|k-1}^c\right\|^2 \le V(e_{k-1|k-1}^c) \le \frac{1}{\underline{p}}\left\|e_{k-1|k-1}^c\right\|^2.
  \end{array}
  \end{equation}
Replacing \eqref{eq:equation17} and \eqref{eq:equation18} into \eqref{eq:equation19} we can write:
  \begin{equation}
 \begin{array}{l}  \label{eq:equation21}
V(e_{k|k}^c)=[\widetilde{A}e_{k-1|k-1}^c+L_{k}+N_{k}]^T P_{k|k}^{-1} [\widetilde{A}e_{k-1|k-1}^c\\+L_{k}+N_{k}].
  \end{array}
  \end{equation}
Applying Lemma 1 of \cite{42} we can write the following inequality:
  \begin{equation}
 \begin{array}{l} \label{eq:equation22}
V(e_{k|k}^c) \le (1-\gamma)V(e_{k-1|k-1}^c) + L_{k}^T P_{k|k}^{-1}[2\widetilde{A}e_{k-1|k-1}^c+L_{k}]^T\\+2N_{k}^T P_{k|k}^{-1}[\widetilde{A}e_{k-1|k-1}^c+L_{k}]+ N_{k}^T P_{k|k}^{-1} N_{k}
  \end{array}
  \end{equation}
Taking conditional expectation $E[V(e_{k|k}^c)|e_{k-1|k-1}^c]$ results in the following inequality: 
  \begin{equation}
 \begin{array}{l} \label{eq:equation23}
\mathbb{E}\left\{ 
V(e_{k|k}^c)|e_{k-1|k-1}^c \right\} - V(e_{k-1|k-1}^c)
\le -\gamma V(e_{k-1|k-1}^c)\\ + \mathbb{E}\left\{  L_{k}^T P_{k|k}^{-1}[2\widetilde{A}e_{k-1|k-1}^c+L_{k}]^T|e_{k-1|k-1}^c\right\}\\ + \mathbb{E}\left\{  2N_{k}^T P_{k|k}^{-1}[\widetilde{A}e_{k-1|k-1}^c+L_{k}]|e_{k-1|k-1}^c\right\} \\ + \mathbb{E}\left\{ N_{k}^T P_{k|k}^{-1} N_{k}|e_{k-1|k-1}^c\right\}.
  \end{array}
  \end{equation}
Each term can be computed as a follows:\\
\begin{itemize}

\item $\mathbb{E}\left\{  L_{k}^T P_{k|k}^{-1}[2\widetilde{A}e_{k-1|k-1}^c+L_{k}]^T|e_{k-1|k-1}^c\right\}$:\\

Using \eqref{eq:equation18}, Assumption \ref{Ass:1} and Assumption \ref{Ass:2} we can obtain the norm of $L_{k}$ as a follows:
  \begin{equation}
 \begin{array}{l} \label{eq:equation24}
\left\|L_{k}\right\| \le \left\|I-K_{k}C_{k}\right\|\left\|\theta_{f}(e_{k-1|k-1}^c)\right\|\\ +\left\|K_{k}\right\|\left\|\phi_{h}(e_{k-1|k-1}^c)\right\| + \left\|K_{k}\right\|\left\|\Delta C_{k}x_{k}\right\| \\ +\left\|I-K_{k}C_{k}\right\|\left\|\Delta A_{k-1}x_{k-1}\right\|  \le (1+\widetilde{k} \widetilde{c})\varepsilon_{\theta} \\ \times \left\|e_{k-1|k-1}^c\right\|^2 +\widetilde{k}\varepsilon_{\phi}\left\|e_{k-1|k-1}^c\right\|^2  +(1+\widetilde{k}\widetilde{c})\delta_{a}+\widetilde{k}\delta_{c}\\ =\varepsilon \left\|e_{k-1|k-1}^c\right\|^2 + \delta,
  \end{array}
  \end{equation}
  
where $\widetilde{k}$ is the upper bound for the Kalman gain $K_{k}$.
Using \eqref{eq:equation24} and under Assumption \ref{Ass:1} we can write:
  \begin{equation}
 \begin{array}{l} \label{eq:equation25}
L_{k}^T P_{k|k}^{-1}[2\widetilde{A}e_{k-1|k-1}^c+L_{k}] \le (\varepsilon \left\|e_{k-1|k-1}^c\right\|^2 + \delta)\frac{1}{\underline{p}}\\\times(2(\widetilde{a}+\widetilde{k}\widetilde{c})\left\|e_{k-1|k-1}^c\right\|+\xi_{k}' \varepsilon  \left\|e_{k-1|k-1}^c\right\| +\xi_{k}' \delta) \\= \kappa_{nl} \left\|e_{k-1|k-1}^c \right\|^3 +\kappa_{u} \left\|e_{k-1|k-1}^c \right\|
  \end{array}
  \end{equation}

\item $\mathbb{E}\left\{  2N_{k}^T P_{k|k}^{-1}[\widetilde{A}e_{k-1|k-1}^c +L_{k}]|e_{k-1|k-1}^c\right\}$:\\
From \eqref{eq:equation18} and considering that $N_{k}$ only depends on $w_{k-1}$ and $v_{k}$, we have:
\begin{equation}
 \begin{array}{l} \label{eq:equation26}
 \mathbb{E}\left\{  2N_{k}^T P_{k|k}^{-1}[\widetilde{A}e_{k-1|k-1}^c +L_{k}]|e_{k-1|k-1}^c\right\}=0
  \end{array}
  \end{equation} 
\item $\mathbb{E}\left\{ N_{k}^T P_{k|k}^{-1} N_{k}|e_{k-1|k-1}^c\right\}$: \\
From \eqref{eq:equation18} and under Assumption \ref{Ass:1}, we have: 
  \begin{equation}
 \begin{array}{l} \label{eq:equation27}
\mathbb E\left\{ 
N_{k}^T P_{k|k}^{-1} N_{k}\right\} \le \frac{1}{\underline{p}} \left\|I-K_{k}C_{k}\right\|^2 \mathbb E \left\{w_{k-1}^Tw_{k-1} \right\} \\
+ \frac{1}{\underline{p}} \left\|K_{k}\right\|^2 \mathbb E \left\{v_{k}^Tv_{k} \right\} \le \frac{1}{\underline{p}} (1+\widetilde{k}\widetilde{c})^2 \widetilde{w}+ \frac{1}{\underline{p}} \widetilde{k}^2 \widetilde{v}\\ = \kappa_{noise}

  \end{array}
  \end{equation}
  
\end{itemize}
Using \eqref{eq:equation22}-\eqref{eq:equation27} we establish the  following inequality:
  \begin{equation}
 \begin{array}{l}  \label{eq:equation28}
\mathbb{E}\left\{ 
V(e_{k|k}^c)|e_{k-1|k-1}^c \right\} - V(e_{k-1|k-1}^c)
\le -\gamma V(e_{k-1|k-1}^c)\\ + \kappa_{nl} \left\|e_{k-1|k-1}^c\right\|^3 +\kappa_{u} \left\|e_{k-1|k-1}^c\right\| +\kappa_{noise}.
  \end{array}
  \end{equation}
Defining: 
  \begin{equation} 
 \begin{array}{l}  \label{eq:equation29}
\left\|e_{k-1|k-1}^c\right\| \le \frac{\gamma}{2\bar{p}\kappa_{nl}}
  \end{array}
  \end{equation}
we can write:
  \begin{equation}
 \begin{array}{l}  \label{eq:equation30}
\kappa_{nl} \left\|e_{k-1|k-1}^c\right\|\left\|e_{k-1|k-1}^c\right\|^2 \le \frac{\gamma}{2\bar{p}} \left\|e_{k-1|k-1}^c\right\|^2.
  \end{array}
  \end{equation}
Hence from \eqref{eq:equation20}, \eqref{eq:equation28}-\eqref{eq:equation30}, Assumption \ref{Ass:1}.(iii) and Lemma 2.1 of \cite{48} we obtain the following inequality:
 \begin{equation}
 \begin{array}{l} \label{eq:equation31}
\mathbb{E}\left\{\left\|e_{k|k}^c\right\|^2\right\} \le \frac{\bar{p}}{\underline{p}}(1-\frac{\gamma}{2})^k\mathbb{E}\left\{\left\|e_{0|0}\right\|^2\right\} + \frac{\xi^{c}}{\underline{p}\gamma}
  \end{array}
  \end{equation}
which guarantees the boundedness of the estimation error under model uncertainty. 
\end{proof}

\begin{lemma} \label{2}
 Consider the stochastic nonlinear system with parameter uncertainties given by \eqref{eq:equation8} and  let Assumptions \ref{Ass:1}-\ref{Ass:2} be satisfied. Then the upper bound for the estimation error in the EKF, $e_{k|k}^c$, satisfies the following bound: 
 \begin{equation}
 \begin{array}{l} \label{eq:equation32}
\mathbb{E}\left\{\left\|e_{k|k}^E \right\|^2\right\} \le \frac{\bar{p}}{\underline{p}}(1-\gamma)^k\mathbb{E}\left\{\left\|e_{0|0}\right\|^2\right\} + \frac{\xi^{E}}{\underline{p}\gamma}
  \end{array}
  \end{equation}
if the initial estimation error satisfies $\left\|e_{0|0}\right\| \le \epsilon $ and $\xi^{E}$ is sufficiently small. Note that $\xi^{E}$ is a function of noise covariance and the parameter uncertainty and $0 <\gamma \le 1$.
\end{lemma}
\begin{proof}
Expanding in Taylor series expansion and ignoring higher order terms, we can compute the prediction and estimation error under the parameter uncertainty as a follows:
 \begin{equation}
 \begin{array}{l} \label{eq:equation33}
e_{k|k-1}^E= A_{k-1}e_{k-1|k-1}^E +w_{k-1}  + \Delta A_{k-1}x_{k-1} 
  \end{array}
  \end{equation}
 \begin{equation}
 \begin{array}{l} \label{eq:equation34}
e_{k|k}^E = \widetilde{A}e_{k-1|k-1}^E +L_{k}+N_{k},
  \end{array}
  \end{equation}
where the estimation error consist of three terms:
  \begin{equation}
 \begin{array}{l} \label{eq:equation35}
 \widetilde{A}=[I-K_{k}C_{k}]A_{k-1}\\
 L_{k}= [I-K_{k}C_{k}]\Delta A_{k-1}x_{k-1}- K_{k}\Delta C_{k}x_{k}\\
 N_{k}=[I-K_{k}C_{k}]w_{k-1}-K_{k}v_{k}.
  \end{array}
  \end{equation}
The rest of the proof is similar to the proof of Lemma \ref{1}, and is omitted.
\end{proof}

\begin{theorem} \label{theorem1}
 Consider the stochastic nonlinear system with parameter uncertainties given by \eqref{eq:equation8} and  let Assumption \ref{Ass:1} and Assumption \ref{Ass:2}  be satisfied. Assuming that
 \begin{enumerate}
    \item The initial estimation error is sufficiently small, {\it i.e.} $\exists \epsilon$:
 \begin{equation}
 \begin{array}{l} \label{eq:equation36}
\left\|e_{0|0}\right\| \le  \epsilon
  \end{array}
  \end{equation}
\item The $\xi^{c}$ defined from Lemma \ref{1} and $\xi^{E}$ defined from Lemma \ref{2} are small enough. 
\item The process noise $w_{k}$ and the measurement noise $v_{k}$ are Gaussian distribution satisfying by \eqref{eq:equation4}-\eqref{eq:equation7}.
\end{enumerate}
Then, the CKF provides better estimation of the system states than the EKF, {\it i.e.}, 
 \begin{equation}
 \begin{array}{l} 
\mathbb{E}\left\{\left\|e_{k|k}^c \right\|\right\} \le \mathbb{E}\left\{\left\|e_{k|k}^E \right\|\right\}
  \end{array}
  \end{equation}
\end{theorem} 
\begin{proof}
We will compare these filters in terms of estimation error. 
Applying Lemma \ref{2} and including the effect of the high order terms in the EKF, equation \eqref{eq:equation32} can be rewritten as a follows:
 \begin{equation}
 \begin{array}{l} \label{eq:equation37}
\mathbb{E}\left\{\left\|e_{k|k}^E \right\|^2\right\} \le \frac{\bar{p}}{\underline{p}}(1-\gamma)^k\mathbb{E}\left\{\left\|e_{0|0}\right\|^2\right\} + \frac{\xi^{E}}{\underline{p}\gamma}\\ + o(e_{k-1|k-1}^E)
  \end{array}
  \end{equation}
Based on conditions of Theorem \ref{theorem1} and using \eqref{eq:equation11} and \eqref{eq:equation37}  we can write:
 \begin{equation}
 \begin{array}{l} \label{eq:equation38}
\lim_{(\left\|e_{0|0}\right\|,\xi^{c})\to (0,0)} \mathbb{E}\left\{\left\|e_{k|k}^c \right\|^2\right\} 	 \le \\ \lim_{(\left\|e_{0|0}\right\|,\xi^{c})\to (0,0)}  [\frac{\bar{p}}{\underline{p}}(1-\frac{\gamma}{2})^k\mathbb{E}\left\{\left\|e_{0|0}\right\|^2\right\} + \frac{\xi^{c}}{\underline{p}\gamma}] \le 
\\ \lim_{(\left\|e_{0|0}\right\|,\xi^{E})\to (0,0)} \mathbb{E}\left\{\left\|e_{k|k}^E \right\|^2\right\} \le
\\ 
\lim_{(\left\|e_{0|0}\right\|,\xi^{E})\to (0,0)} [\frac{\bar{p}}{\underline{p}}(1-\gamma)^k\mathbb{E}\left\{\left\|e_{0|0}\right\|^2\right\} + \frac{\xi^{E}}{\underline{p}\gamma}\\ + o(e_{k-1|k-1}^E)]
  \end{array}
  \end{equation}
 If the nonlinearities in the model are mild, then by Assumption \ref{Ass:3}, we can assume the error incurred in neglecting the high order terms ($o(e_{k-1|k-1}^E)$) is small enough. So, the CKF and EKF will result similar performance, {\it i.e.}, 
 \begin{equation}
 \begin{array}{l} \label{eq:equation38}
\mathbb{E}\left\{\left\|e_{k|k}^c \right\|\right\} \cong \mathbb{E}\left\{\left\|e_{k|k}^E \right\|\right\}
  \end{array}
  \end{equation}
   However, for the highly nonlinear systems, as long as Assumption \ref{Ass:3} is not satisfied for the high order terms, then  we can not ignore the higher order terms and the error of neglecting the high order term will enlarge the upper bound of convergence in the EKF method. So, CKF provides smaller upper bound for the estimation error rather than the EKF under the proposed region (process and measurement noise, uncertainty parameter and estimation error become small enough), {\it i.e.},
    \begin{equation}
 \begin{array}{l} \label{eq:equation38}
\mathbb{E}\left\{\left\|e_{k|k}^c \right\|\right\} < \mathbb{E}\left\{\left\|e_{k|k}^E \right\|\right\}
  \end{array}
  \end{equation}
This completes the proof.
\end{proof}
\section{Attack Model}  
 In this paper, we considered the following four
types of attacks \cite{33,43,44}.
\begin{enumerate}
\item Random attack: An adversary 
manipulates the sensor readings. This type of attack can be modelled as follows: 
\begin{equation} \label{eq:equation39}
y_{k}^a = h(x_{k},u_{k})+v_{k}+ \delta_{k}
\end{equation}
Where $y_{k}^a$, $v_{k}$ and $\delta_{k}$ are corrupted measurement, measurement noise and random attack vector injected to the measurement, respectively. In this paper, we assume that the random attack is bounded by arbitrary energy signal as a follows:
\begin{equation} \label{eq:equation40}
\left\|\delta_{k}\right\|\le \widetilde{\delta}
\end{equation}
\item Denial of Service (DoS) attack: An adversary prevents communication with the sensors by jamming the communication channels, thus flooding packets in the network. The DoS attack can be a barrier in sensor data transfer, control data transfer, or both. In this paper, we model the DoS attack as the lack of the updated measurements to the estimator: 
\begin{equation} 
y_{t}^a=
\left\{
                \begin{array} {l} \label{eq:equation41}

                h(x_{t1},u_{t1})+ v_{t} \ \ \ \ \ \ \ t\in (t_{1},t_{2}] \\
                h(x_{t},u_{t})+v_{t} \ \ \ \ \ \ \ t\notin (t_{1},t_{2}],
                \end{array}
\right.
\end{equation}
where the measurement is keep unchanged for $t\in (t_{1},t_{2}]$ and the updated measurement cannot be sent to the estimator .
\item Replay attack: A replay attack is a form of network attack in memorizes a string of sensor readings and transmit then at a later time. A replay attack on the output $i$ at $t\in [t_{1},t_{2}]$ can be modelled as follows:
\begin{equation} \label{eq:equation42}
y_{i,t}^a=h_{i}(x_{t-\Delta T},u_{t-\Delta T}) + v_{t} \ \ \ \ \ \ \ t\in [t_{1},t_{2}]
\end{equation}
where $\Delta T=t_{2}-t_{1}$ 
\item False data injection (FDI) attack: In this case, an attacker knows the system model, including the parameters $A$, $B$, $C$, $Q$, $R$, and gain $K$ and injects false sensor measurements in order to confuse the system when implementing a Kalman filter estimator with the $\chi^{2}$-detector \cite{45}. To clarify the idea, let's consider the following system \cite{46}:
\begin{equation} 
\left\{
                \begin{array} {l} \label{eq:equation43}
                x_{k} = f(x_{k-1},u_{k-1})+w_{k-1}\\
                y_{k} = h(x_{k},u_{k})+v_{k} \\
                 y_{k}^a = y_{k}+ a_{k}
                \end{array}
\right.
\end{equation}
where $y_{k}$, $y_{k}^a$ and $a_{k}$ are the measured output, corrupted measurement and FDI attack generated by the attacker, respectively. The FDI attack can be described as:
\begin{equation} \label{eq:equation44}
a_{k}=-y_{k}+\eta_{k}
\end{equation}
where $-y_{k}$ will cancel the original signal and $\eta$ is assumed to be an arbitrary bounded energy signal (detection threshold) with the following characteristic:
\begin{equation} \label{eq:equation45}
\left\|\eta_{k}\right\| \le \widetilde{\eta}
\end{equation}
\end{enumerate}

\begin{lemma} \label{3}
 Consider the nonlinear stochastic system under random attack of the form \eqref{eq:equation39}. Let's Assumption \ref{Ass:1} and Assumption \ref{Ass:2} hold. Then, if the initial estimation error satisfies $\left\|e_{0|0}\right\| \le \epsilon $ and $\psi^c$ is sufficiently small, then the estimation error $e_{k|k}^c$ using the CKF algorithm satisfies the following bound:
 \begin{equation}
 \begin{array}{l} \label{eq:equation46}
\mathbb{E}\left\{\left\|e_{k|k}^c\right\|^2\right\} \le \frac{\bar{p}}{\underline{p}}(1-\frac{\gamma}{2})^k\mathbb{E}\left\{\left\|e_{0|0}\right\|^2\right\} + \frac{\psi^c}{\underline{p}\gamma}
  \end{array}
  \end{equation}
Note that $\psi^c$ is a function of noise covariance and random attack and $0 <\gamma \le 1$.  
 \end{lemma}
\begin{proof}
By considering the random attack presented in measurement, the prediction error and the estimation error are given by:
\begin{equation}
 \begin{array}{l} \label{eq:equation47}
e_{k|k-1}^c=x_{k}-\hat{x}^c_{k|k-1}=f(x_{k-1},u_{k-1})+w_{k-1}\\ - \frac{1}{m}\sum\limits_{m=1}^{m} f(X_{i,k-1|k-1},u_{k-1})=A_{k-1}e_{k-1|k-1}^c+\\ \theta_{f}(|e_{k-1|k-1}^c|)+w_{k-1} 
  \end{array}
  \end{equation}
  \begin{equation}
 \begin{array}{l} \label{eq:equation48}
e_{k|k}^c=x_{k}-\hat{x}_{k|k}=f(x_{k-1},u_{k-1}) +w_{k-1}\ -\hat{x}^c_{k|k-1} \\ -K_{k}(y_{k}-h(\hat{x}^c_{k|k-1}))=A_{k-1}e_{k-1|k-1}^c +\theta_{f}(|e_{k-1|k-1}^c|)\\ +w_{k-1}-K_{k}(h(x_{k},u_{k})+v_{k}+ \delta_{k}\\-\left\{ 
\frac{1}{m}\sum\limits_{m=1}^{m} h(\hat{x}^c_{k-1|k-1}+\Delta x_{k-1|k-1}^i,u_{k-1})\right\})
  \end{array}
  \end{equation}
 Thus 
 \begin{equation}
 \begin{array}{l} \label{eq:equation49}
e_{k|k}^c=\widetilde{A}e_{k-1|k-1}^c+L_{k}+N_{k}
\end{array}
\end{equation}

Where the estimation error consist of three terms:
  \begin{equation}
 \begin{array}{l} \label{eq:equation50}
 \widetilde{A}=[I-K_{k}C_{k}]A_{k-1}\\
 L_{k}=[I-K_{k}C_{k}]\theta_{f}(e_{k-1|k-1}^c)-K_{k}\phi_{h}(e_{k-1|k-1}^c)- K_{k} \delta_{k} \\
 N_{k}=[I-K_{k}C_{k}]w_{k-1}-K_{k}v_{k}
  \end{array}
  \end{equation}

We will follow all steps from \eqref{eq:equation19} to \eqref{eq:equation22} which yields the inequality \eqref{eq:equation23}. All terms in the inequality \eqref{eq:equation23} will be the same except the second term which can be rewritten as a following: 
  \begin{equation}
 \begin{array}{l} \label{eq:equation51}
\left\|L_{k}\right\| \le \left\|I-K_{k}C_{k}\right\| \left\|\theta_{f}(e_{k-1|k-1}^c)\right\| +\left\|K_{k}\right\|\left\|\phi_{h}(e_{k-1|k-1}^c)\right\| \\ + \left\|K_{k}\right\|\left\|\delta_{k} \right\| \le (1+\widetilde{k}\widetilde{c})\varepsilon_{\theta} \left\|e_{k-1|k-1}^c\right\|^2 +\widetilde{k}\varepsilon_{\phi} \left\|e_{k-1|k-1}^c\right\|^2 \\ + \widetilde{k}\widetilde{\delta}=\varepsilon \left\|e_{k-1|k-1}^c\right\|^2 + \varepsilon_{\delta}
\end{array}
\end{equation}
Using \eqref{eq:equation18} and under Assumption \ref{Ass:1} and Assumption \ref{Ass:2} we have:
  \begin{equation}
 \begin{array}{l} \label{eq:equation52}
L_{k}^T P_{k|k}^{-1}[2\widetilde{A}e_{k-1|k-1}^c+L_{k}] \le (\varepsilon \left\|e_{k-1|k-1}^c\right\|^2 + \varepsilon_{\delta})\frac{1}{\underline{p}}\\\times(2(\widetilde{a}+\widetilde{k}\widetilde{c})\left\|e_{k-1|k-1}^c\right\|+\xi_{k}' \varepsilon  \left\|e_{k-1|k-1}^c\right\| +\xi_{k}' \varepsilon_{\delta}) \\= \kappa_{nl} \left\|e_{k-1|k-1}^c\right\|^3 +\kappa_{a} \left\|e_{k-1|k-1}^c\right\|
\end{array}
\end{equation}
Using \eqref{eq:equation23}, \eqref{eq:equation26}, \eqref{eq:equation27} and \eqref{eq:equation52} the following inequality can be established:
  \begin{equation}
 \begin{array}{l} \label{eq:equation53}
\mathbb{E}\left\{ 
V(e_{k|k}^c)|e_{k-1|k-1}^c \right\} - V(e_{k-1|k-1}^c)
\le -\gamma V(e_{k-1|k-1}^c)\\ + \kappa_{nl} \left\|e_{k-1|k-1}^c\right\|^3 +\kappa_{a} \left\|e_{k-1|k-1}^c\right\| +\kappa_{noise}
  \end{array}
  \end{equation}
Thus, from \eqref{eq:equation20}, \eqref{eq:equation29}, \eqref{eq:equation30}, \eqref{eq:equation53}, Assumption \ref{Ass:1}.(iii) and Lemma 2.1 of \cite{48} we can obtain the following inequality:
\begin{equation} 
 \begin{array}{l} \label{eq:equation54}
\mathbb{E}\left\{\left\|e_{k|k}^c\right\|^2\right\} \le \frac{\bar{p}}{\underline{p}}(1-\frac{\gamma}{2})^k\mathbb{E}\left\{\left\|e_{0|0}\right\|^2\right\} + \frac{\psi^c}{\underline{p}\gamma}
  \end{array}
  \end{equation}
\end{proof}

\begin{lemma} \label{4}
 Consider the nonlinear stochastic system under random attack as stated in \eqref{eq:equation39}. Let's Assumptions \ref{Ass:1}-\ref{Ass:2} hold. 
 Then, there exist upper bound for the estimation error $e_{k|k}^E$ using the EKF algorithm as a follows:
 \begin{equation}
 \begin{array}{l} \label{eq:equation55}
\mathbb{E}\left\{\left\|e_{k|k}^E\right\|^2\right\} \le \frac{\bar{p}}{\underline{p}}(1-\gamma)^k\mathbb{E}\left\{\left\|e_{0|0}\right\|^2\right\} + \frac{\psi^E}{\underline{p}\gamma} 
  \end{array}
  \end{equation}
if the initial estimation error satisfies $\left\|e_{0|0}\right\| \le \epsilon $ and $\psi^E$ is sufficiently small. Note that $\psi^E$ is a function of noise covariance and random attack and $0 <\gamma \le 1$.
 \end{lemma}
\begin{proof}
The proof is based on Lemma \ref{3} and analogous with Lemma \ref{2}, hence omitted.
\end{proof}

\begin{theorem} \label{theorem2}
Consider the nonlinear stochastic system under random attack given by \eqref{eq:equation39}. Let's Assumption \ref{Ass:1} and Assumption \ref{Ass:2} hold. Assuming that
 \begin{enumerate}
    \item The initial estimation error is sufficiently small, {\it i.e.} $\exists \epsilon$:
 \begin{equation}
 \begin{array}{l} \label{eq:equation36}
\left\|e_{0|0}\right\| \le  \epsilon
  \end{array}
  \end{equation}
Where $\epsilon$ is sufficiently Small. 
\item The $\psi^{c}$ defined from Lemma \ref{3} and $\psi^{E}$ defined from Lemma \ref{4} are small enough. 
\item The process noise $w_{k}$ and the measurement noise $v_{k}$ are Gaussian distribution satisfying by \eqref{eq:equation4}-\eqref{eq:equation7}.
\end{enumerate}
Then, the CKF provides better estimation of the system states than the EKF, {\it i.e.},
 \begin{equation}
 \begin{array}{l} 
\mathbb{E}\left\{\left\|e_{k|k}^c \right\|\right\} \le \mathbb{E}\left\{\left\|e_{k|k}^E \right\|\right\}
  \end{array}
  \end{equation}
\end{theorem}

\begin{proof}
The proof follows similar steps used in the proof of Theorem \ref{theorem1}, hence omitted.
\end{proof}

\section{ATTACK/FAILURE DETECTOR}
When reading data is available, the nonlinear filters can be used to generate the estimated states. A detector is utilized to compare the actual sensor readings and the estimated output with the purpose of detection. If the difference between these two is above a threshold, an alarm is triggered to notify the presence of the attacker. In this paper, two types of detectors are employed for the detection of potential attacks.
\begin{enumerate}
\item $\chi^{2}$-detector: The $\chi^{2}$-detector is a residue-based detector which makes a decision based on the sum of squared residues $z_{k+1}$. The residue $z_{k+1}$ at time $k+1$ is defined \cite{33}:
\begin{equation} \label{eq:equation57}
z(t+1) = y(t+1)-\hat{y}(t+1)
\end{equation}

Then, the $\chi^{2}$-detector test is normalized by the covariance matrix of $z(t)$ as follows:
\begin{equation} \label{eq:equation58}
g(t) = z(t)^TB(t)z(t)
\end{equation}
In summary, the $\chi^{2}$-detector compares g(t) with the threshold obtained using
the $\chi^{2}$-detector-table \cite{47} to identify an attack. The only limitation encountered in the $\chi^{2}$-detector is detection of FDI attack, which Euclidean detector can be used for this purpose.
\item Euclidean detector: The Euclidean detector was proposed by Manandhar et al. \cite{33} to overcome the limitations of $\chi^{2}$-detector. We will show that FDI attacks may fail to get detected by $\chi^{2}$-detector. The Euclidean detector is defined by the following equation:
\begin{equation}  \label{eq:equation59}
 d=\sqrt{(y_{obs,1}-y_{est,1})^2 + \dots + (y_{obs,n}-y_{est,n})^2}
\end{equation}
where $y_{obs}=[y_{obs,1},y_{obs,2}, \dots , y_{obs,n}]$ represents measurement vector and $y_{est}=[y_{est,1},y_{est,2}, \dots , y_{est,n}]$ is the estimated output vector. 
\end{enumerate}

\section{Case Study}
In this section, we compare the performance of the nonlinear filters in the presence of modelling error and cyber attacks on the synchronous machine shown in Fig. \ref{fig:figure 1}.  Equations \eqref{eq:equation1} and \eqref{eq:equation2} represents the fourth-order nonlinear synchronous machine state-space model \cite{35 ,36}.    
\begin{align} \label{eq:equation1}
&[\delta \ \ \Delta \omega \ \ e_{q}' \ \ e_{d}']^T =[x_{1} \ \ x_{2}\ \ x_{3}\ \ x_{4}]^T 
\nonumber \\
&[T_{m}\ \ E_{fd}]^T =[u_{1} \ \ u_{2}]^T  \nonumber
\\
&\left\{
                \begin{array} {l}
                \dot{x_{1}}= \omega _{0}x_{2} \\
               \dot{x_{2}} = \frac{1}{J}(T_{m}-T_{e}-Dx_{2})  \\
               \dot{x_{3}} = \frac{1}{T_{do}'}(E_{fd}-x_{3}-(x_{d}-x_{d}')i_{d}) \\
                 \dot{x_{4}} = \frac{1}{T_{qo}'}(-x_{4}-(x_{q}-x_{q}')i_{q})\\
                \end{array}
\right.
\\ 
  &              \begin{array} {l}
                y_{1}(t)=\frac{V_{t}}{x_{d}'}(x_{3})\sin(x_{1})+\frac{V_{t}^2}{2}(\frac{1}{x_{q}}-\frac{1}{x_{q}'})\sin(2x_{1})\\
                \end{array}.\nonumber
\end{align}

Where $\omega_{0}=2\pi f_{0}$ is the nominal synchronous speed (elec.
rad/s), $\omega$ the rotor speed (pu), $T_{m}$ the mechanical input torque
(pu), $T_{e}$ the air-gap torque or electrical output power (pu), $E_{fd}$
the exciter output voltage or the field voltage as seen from the armature (pu), and $\delta$ the rotor angle in (elec.rad). The air-gap torque or electrical output power $T_{e}$ and the d-axis and q-axis currents ($i_{d}$,$i_{q})$ can be computed by following:
\begin{equation} 
\label{eq:equation2}
\left\{
                \begin{array} {l}
                T_{e}=\frac{V_{t}}{x_{d}'}(x_{3})\sin(x_{1})+\frac{V_{t}^2}{2}(\frac{1}{x_{q}}-\frac{1}{x_{q}'})\sin(2x_{1})\\
                i_{d}=\frac{x_{3}-V_{t}\cos(x_{1})}{x_{d}'}\\
                i_{q}=\frac{V_{t}\sin({x_{1})}}{x_{q}}\\
                \end{array}
\right.
\end{equation}
More details about parameters and variables are described in Table 1.\\
\begin{figure}[h]
  \begin{center}
\includegraphics[width=0.45\textwidth]{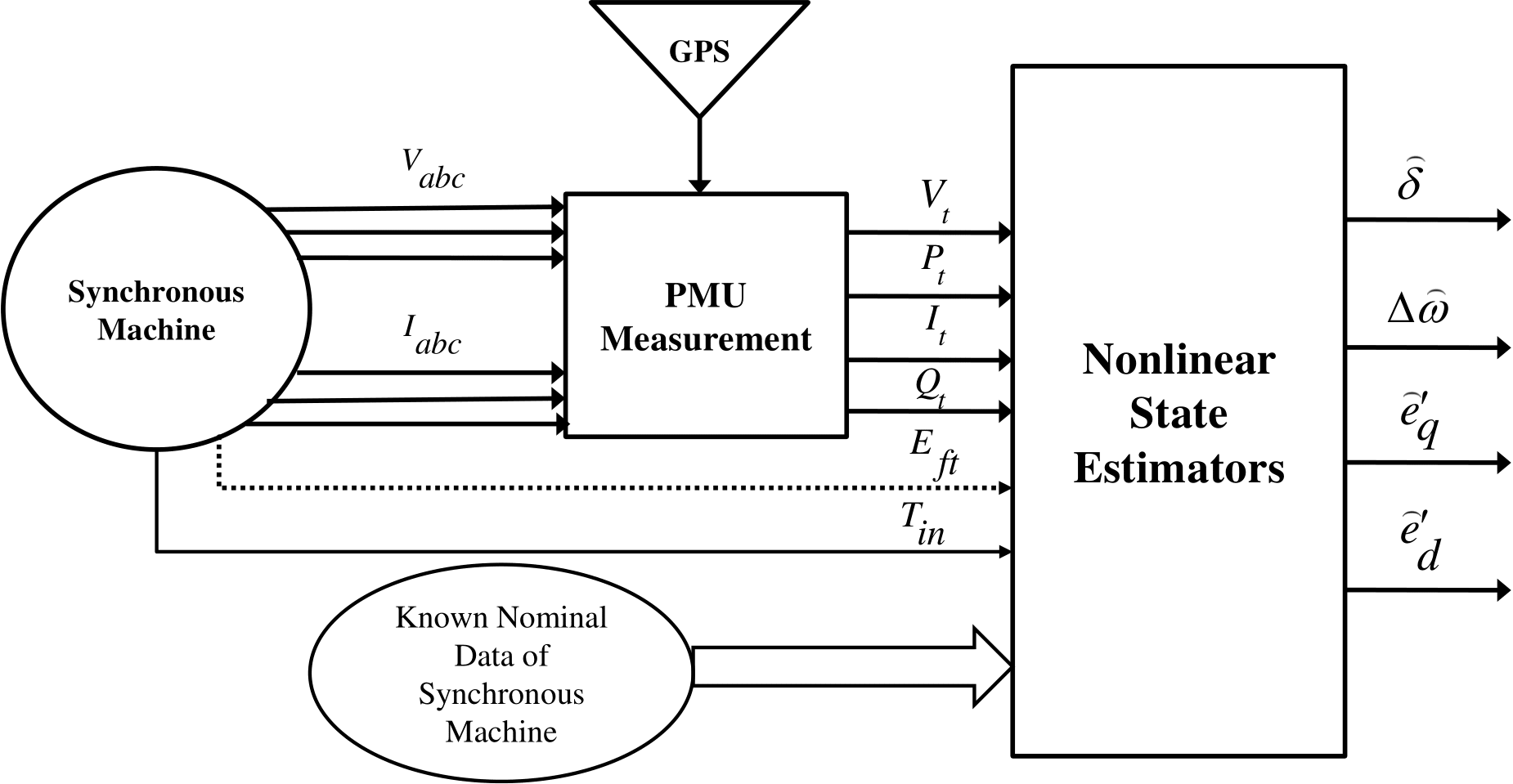}
  \end{center}
\caption{Diagram of nonlinear state estimator for a synchronous machine.} \label{fig:figure 1}
\end{figure}  
For the synchronous machine we consider both parameter uncertainty and cyber-attacks against the measurement.
Fig.\ref{fig:figure 2} shows the block diagram of the estimation and detection mechanism for the synchronous machine. Our scheme consists of five components namely, synchronous machine, PMU measurement, attacker, nonlinear state estimators and detectors. 
 The PMU device measures input voltage and current samples from the synchronous machine and transmits data packets to the estimator where an attacker may manipulate the transmitted data. A detector monitors the system behavior and identifies the presence of the attacker.
\begin{figure}[h]
  \begin{center}
\includegraphics[width=0.47\textwidth]{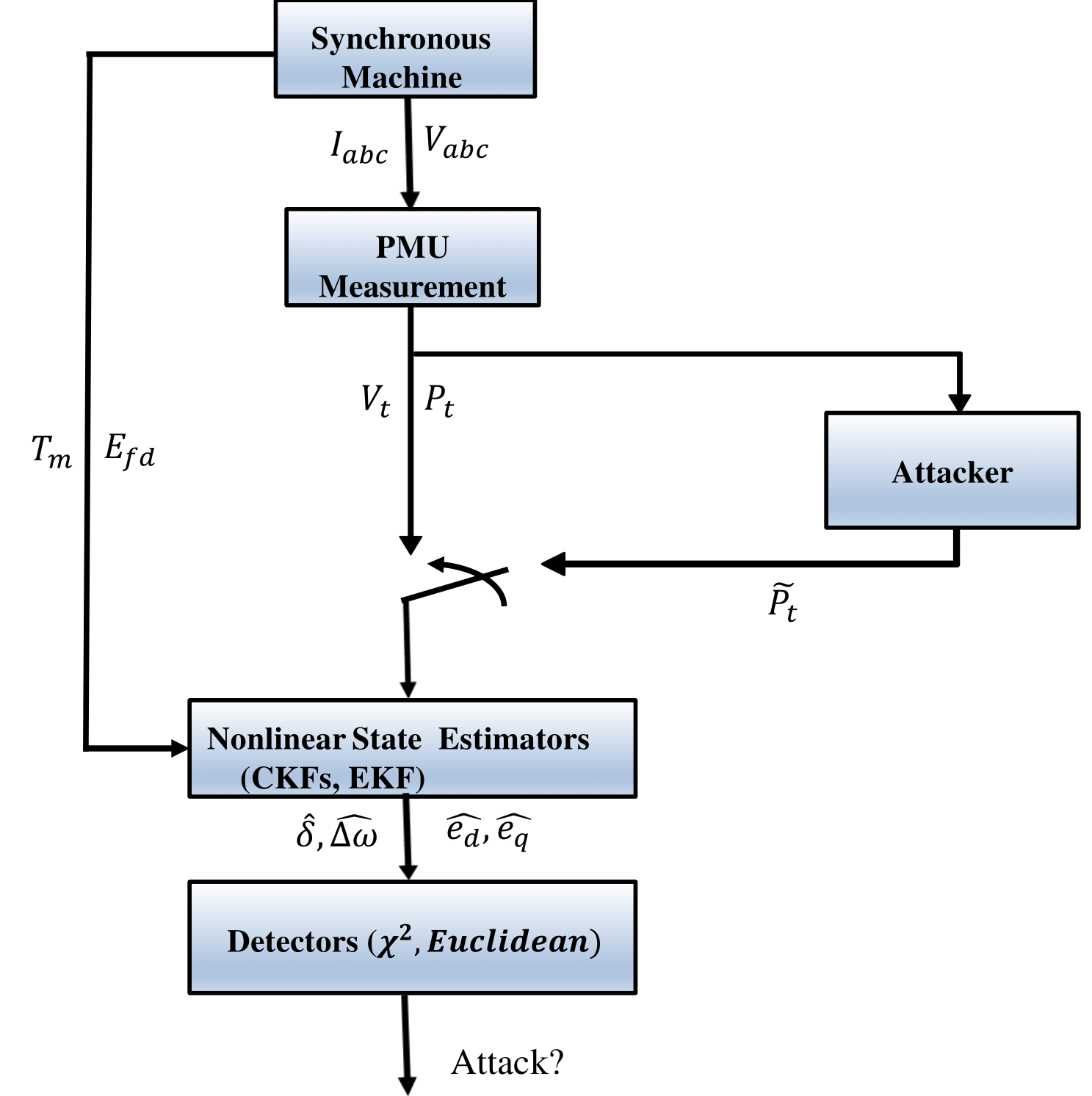}
  \end{center}
\caption{Block diagram of estimation and detection mechanism for the  synchronous machine.}
 \label{fig:figure 2}
\end{figure} 

In order to test the aforementioned estimators performance on the synchronous machine, we consider five scenarios:\\
Scenario 1: normal conditions (no modelling error and no attack) \\
Scenario 2: noisy measurements \\
Scenario 3: model uncertainty \\
Scenario 4: cyber-attack \\
Scenario 5: cyber-attack detection. \\
We implement discrete-time EKF, CKF and SCKF algorithms in MATLAB. 
Table 1. defines the parameters of the generator model. The initial value of the states, estimated states and state covariance matrix considered are: $x_{obs}=[0.4;0;0;0]$, $x_{est}=[0.4;0;0;0]$ and $P_{0}=diag([10^2, 10^2, 10^2, 10^2])$, respectively. In the simulations, two input $T_{m}=0.8$ (\textbf{constant}) and $E_{fd}=$ \textbf{step} with initial value of 2.11 and final value of 2.32 are applied to the system. The noise added to the process and measurement consists of white-noise sequences with covariance $w_{k}\sim (0,Q_{k})=(0,0.001^2\times I_{4\times4})$ and $v_{k}\sim (0,R_{k})=(0,0.01^2\times I)$.
\subsection{\textbf{Estimator's performance under normal conditions}}
In this section, we compare the estimation provided by the CKF, SCKF and EKF ignoring the effect of modelling error and cyber attack.   
Fig. \ref{fig:figure 3} shows the performance of EKF, CKF, and SCKF in normal conditions. Considering the nonlinear model, all estimated states converge to real states.  
\begin{figure}[h]
  \begin{center}
\includegraphics[scale=0.57]{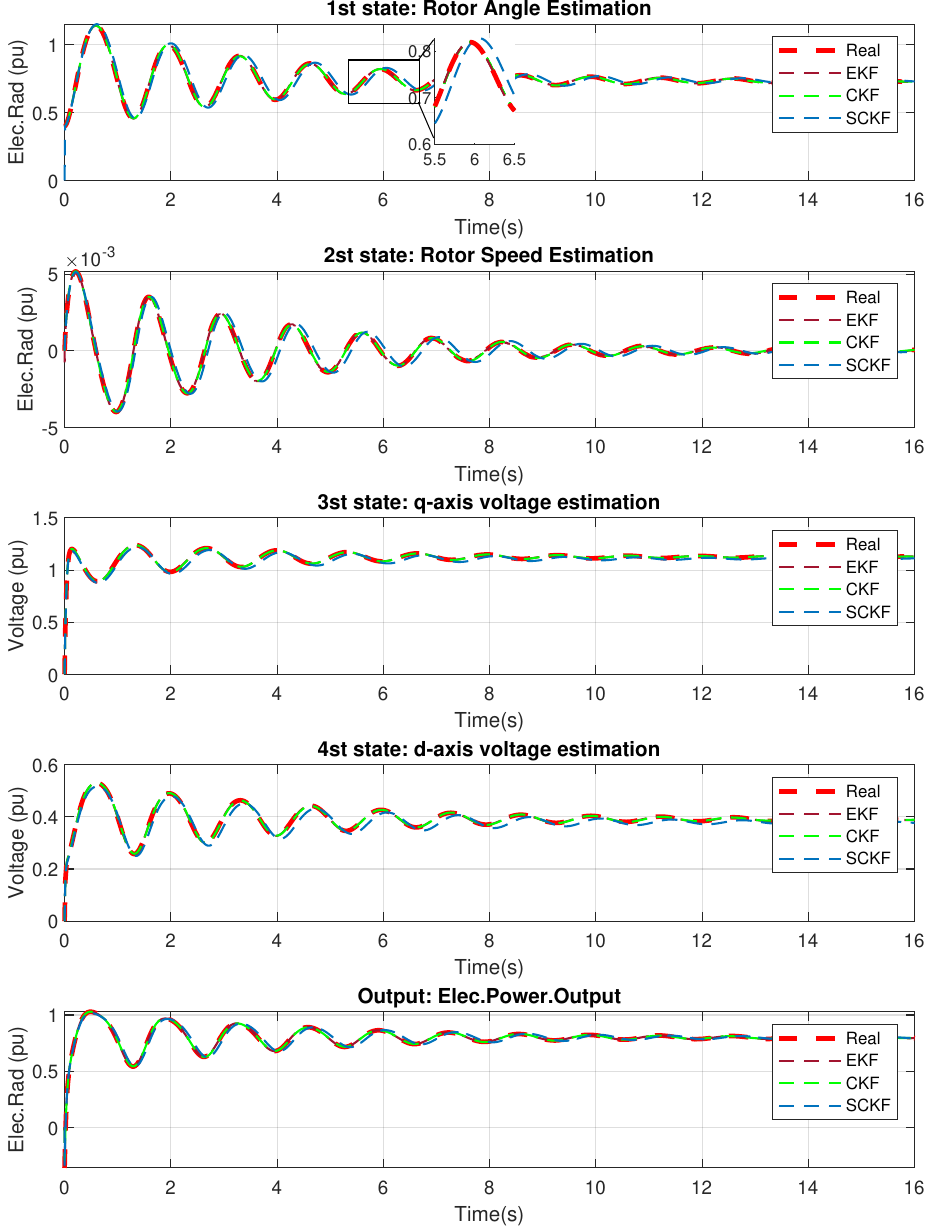}
  \end{center}
\caption{State estimation results under normal conditions.}
 \label{fig:figure 3}
\end{figure} 
\subsection{\textbf{Estimator's Performance with Noisy Measurements}}
In this part, the process and measurement covariance are set as $w_{k}\sim (0,Q_{k})=(0,0.001^2\times I_{4\times4})$ and $v_{k}\sim (0,R_{k})=(0,0.5^2\times I)$. Fig. \ref{fig:figure 4} shows the state estimation results under corrupted measurement which CKF and SCKF have better performance in terms of state tracking and robustness.  
\begin{figure}[h]
  \begin{center}
\includegraphics[scale=0.57]{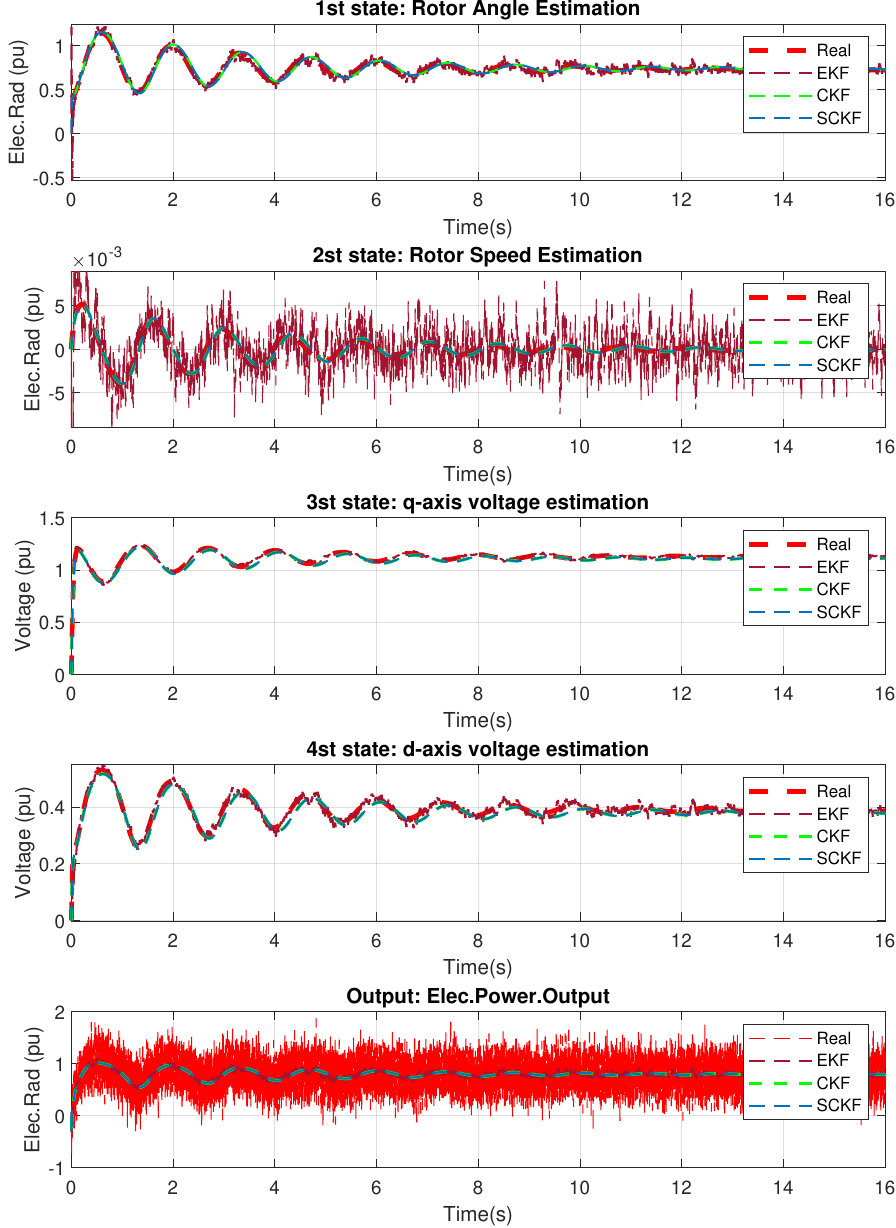}
  \end{center}
\caption{State estimation results with noisy measurements.}
 \label{fig:figure 4}
\end{figure} 
\subsection{\textbf{Estimator's performance under model uncertainty}}
In this section, we assume that the reactance ($x_{d}',x_{q}'$) in \eqref{eq:equation1} and \eqref{eq:equation2} change abruptly from 0.375 to 0.475 abruptly at $t=2.5s$ due to a short-circuit fault in the transmission line. The estimated states by the EKF, CKF, and SCKF are shown in Fig. \ref{fig:figure 5}. It can be seen that the estimation results using the EKF deviate from original states at $t=2.5s$ due to sudden change in the reactance. However, 
CKF and SCKF provide estimates with high accuracy. The example  is representative of the robustness of CKF and SCKF under parameter uncertainty. Also, the estimates provided by the CKF and SCKF do not dependent on initial value of the estimates being close to the real state. By contrast, the EKF is very sensitive to the initial value of the rotor angle. 
\begin{figure}[h]
  \begin{center}
\includegraphics[scale=0.58]{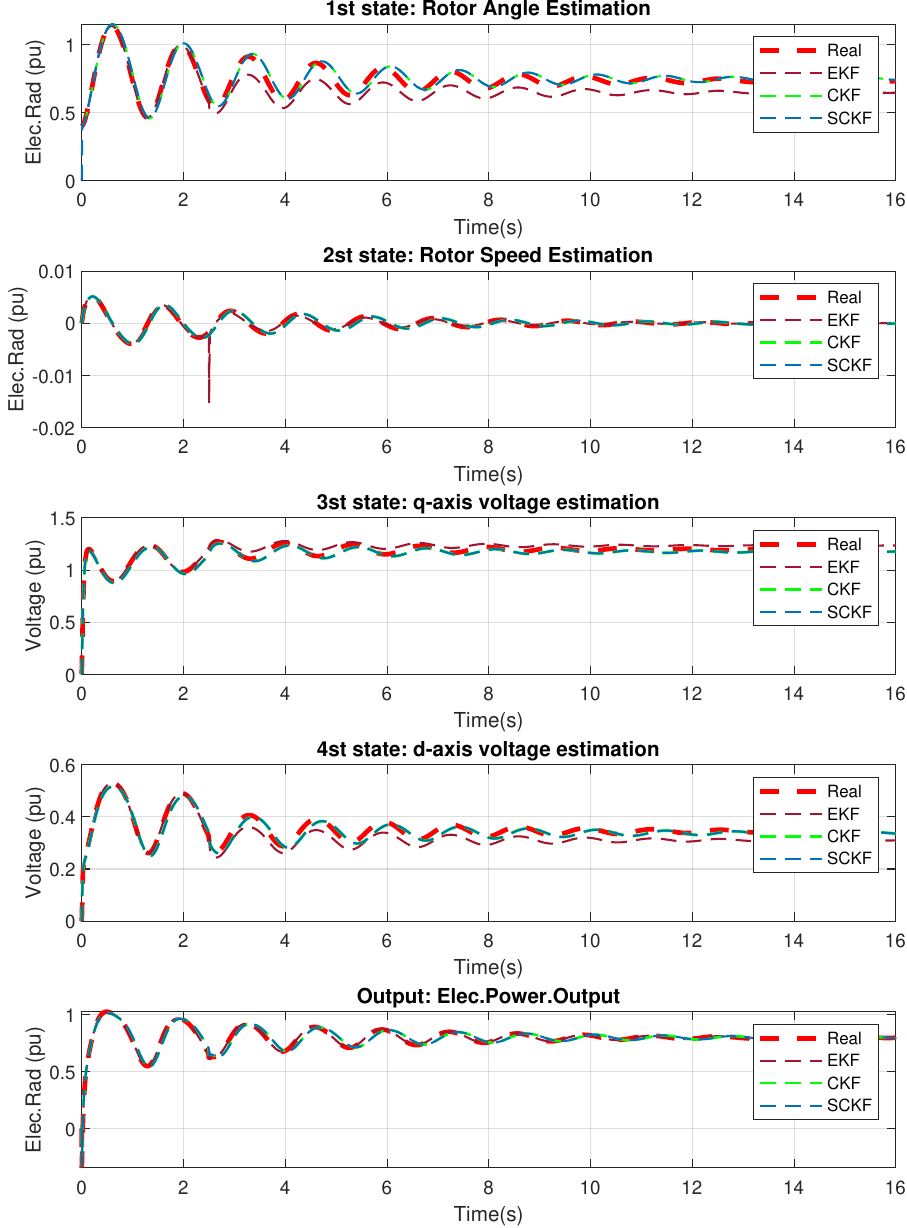}
  \end{center}
\caption{State estimation results under model uncertainty.}
 \label{fig:figure 5}
\end{figure} 
\subsection{\textbf{Estimator's Performance under Cyber Attacks}}
In this scenario, four types of attacks are applied to the system and the performance of estimators are compared.\\
For the modelling of random attack, we assume that $\delta_{t}=0.1sin(2\pi 60t)$ is added to the measurement at $t=0$.\\ Also, replay attack can be added to the measurement as a delay of the form $y_{i}^a(t)=y_{i}(t-0.3)$.\\
In order to model the DoS attack, we assume that the measurement is kept unchanged for $t \in [0.2s,1.8s]$.\\
For the modelling of FDI attack, we assume that the attacker knows the Kalman gain and manipulates the Kalman gain in order to corrupt the dynamic state estimation without the attack being detected. We model the FDI attack by modifying the actual Kalman gain ($K^a=K\times diag([0.05,0,0,0])$) and manipulating the measurement ($y^a(t)=y(t)+0.05$).  \\
Fig. \ref{fig:figure 6}, Fig. \ref{fig:figure 7}, Fig. \ref{fig:figure 8} and Fig. \ref{fig:figure 10} show state estimation results under random, DoS, replay, and FDI attacks. 
Both the CKF and SCKF closely track the real states with high accuracy, while the EKF does not. Also, it can be seen in the figures that there is a difference between estimated outputs using CKF and SCKF and measurement which can be used  in the attack detection.
\subsection{\textbf{Cyber-attack detection}}
In this scenario, we consider two cases: Firstly, the performance of different estimators using the $\chi^{2}$-detector for random, DoS and replay attacks.  Fig. \ref{fig:figure 9}, Fig. \ref{fig:figure 11} and Fig. \ref{fig:figure 12} show the normalized innovation ratio $g(t)$ obtained from $\chi^{2}$-detector under the aforementioned cyber attacks. It is clear that $g(t)$ passes the threshold in using both CKF and SCKF and triggers an alarm indicating the presence of an attack, while the EKF is not successful in doing so. Secondly, to handle the FDI attack, the Euclidean detector was employed as an alternative detector to overcome the limitation of $\chi^{2}$-detector. As can be seen in Fig. \ref{fig:figure 13}, the normalized innovation ratio $g(t)$ obtained from $\chi^{2}$-detector remains under the threshold, 
however, $g(t)$ obtained from the Euclidean detector is above the threshold and triggers an
alarm.
\begin{figure}[h]
  \begin{center}
\includegraphics[scale=0.58]{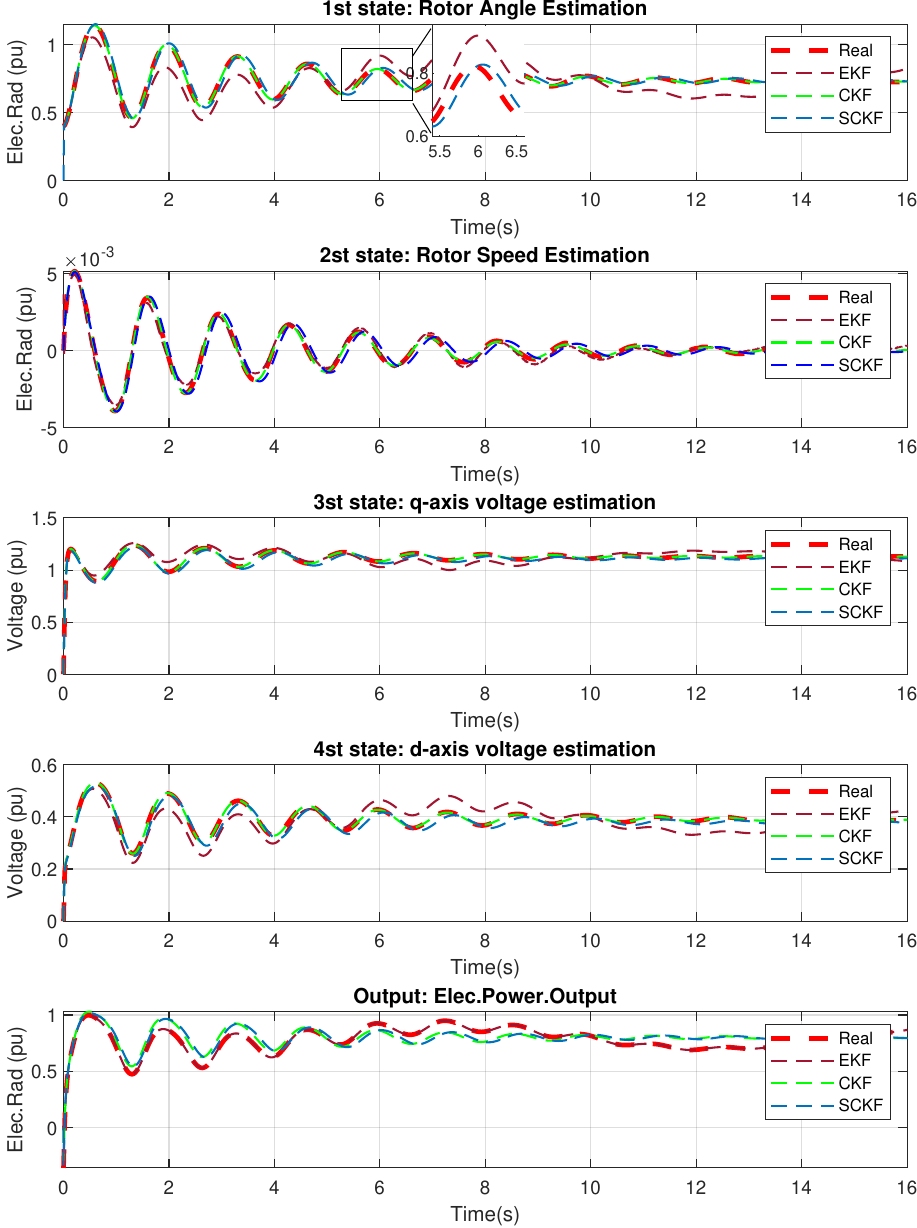}
  \end{center}
\caption{State estimation results under random attack.}
 \label{fig:figure 6}
\end{figure} 

\begin{figure}[h]
  \begin{center}
\includegraphics[scale=0.58]{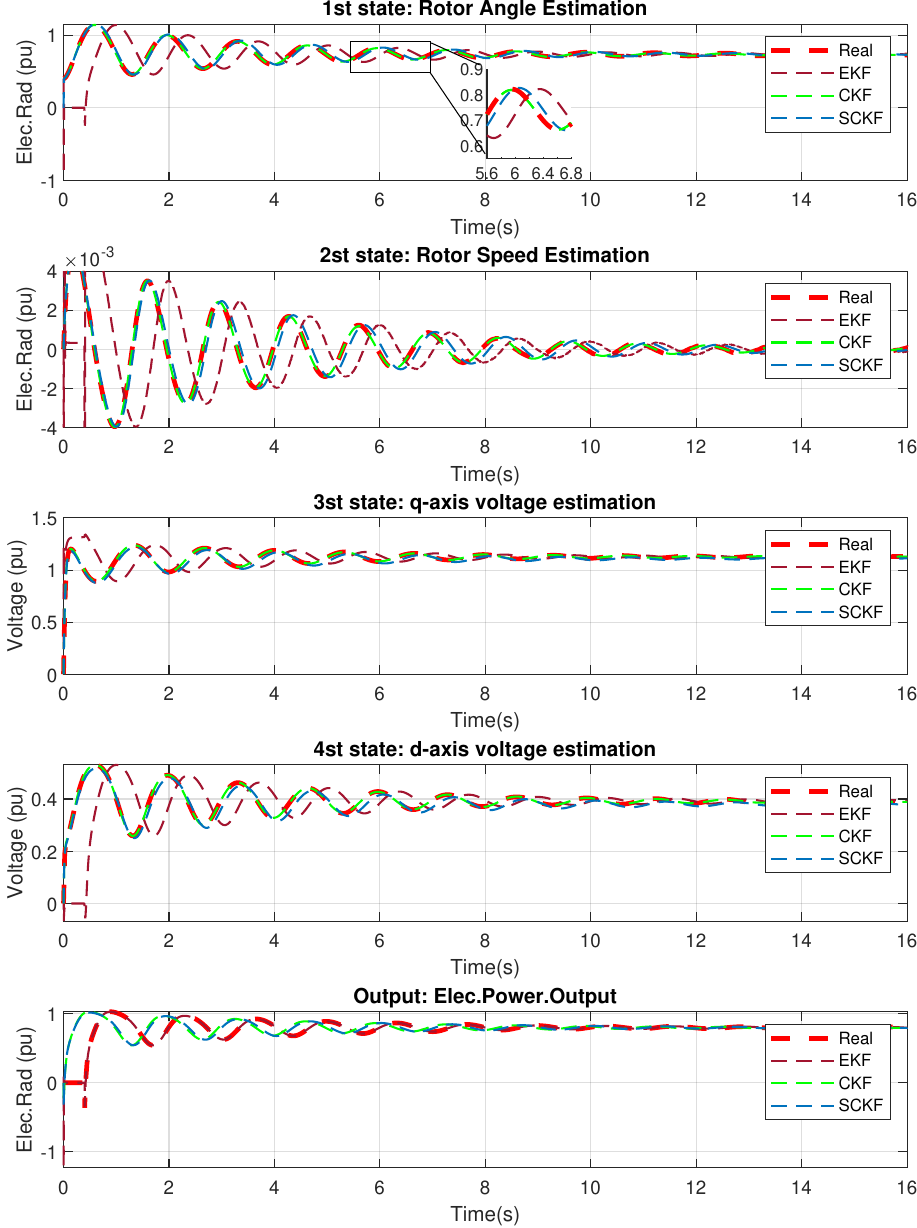}
  \end{center}
\caption{State estimation results under replay attack.}
 \label{fig:figure 7}
\end{figure} 
\begin{figure}[h]
  \begin{center}\includegraphics[scale=0.58]{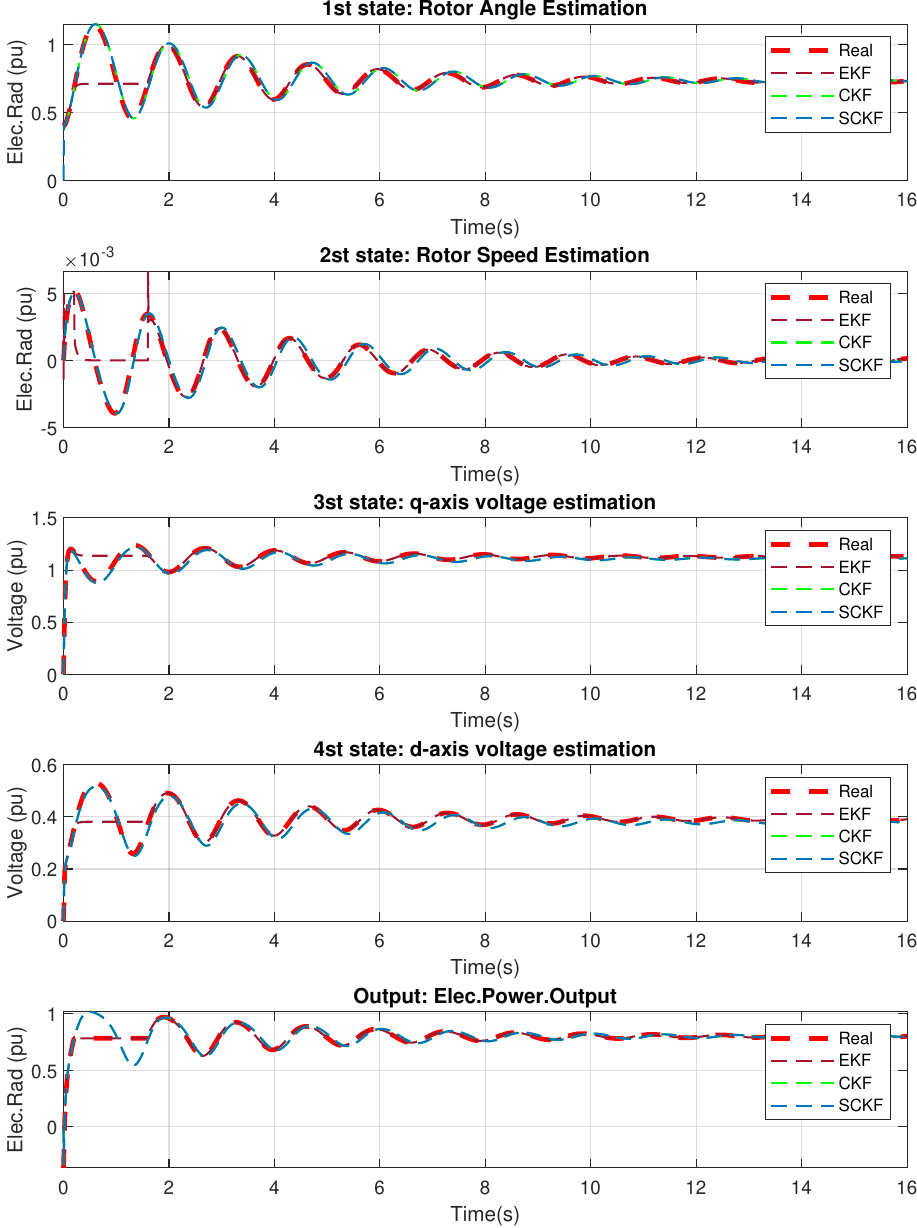}
  \end{center}
\caption{State estimation results under DoS attack.}
 \label{fig:figure 8}
\end{figure} 
\begin{figure}[h]
  \begin{center}
\includegraphics[scale=0.46]{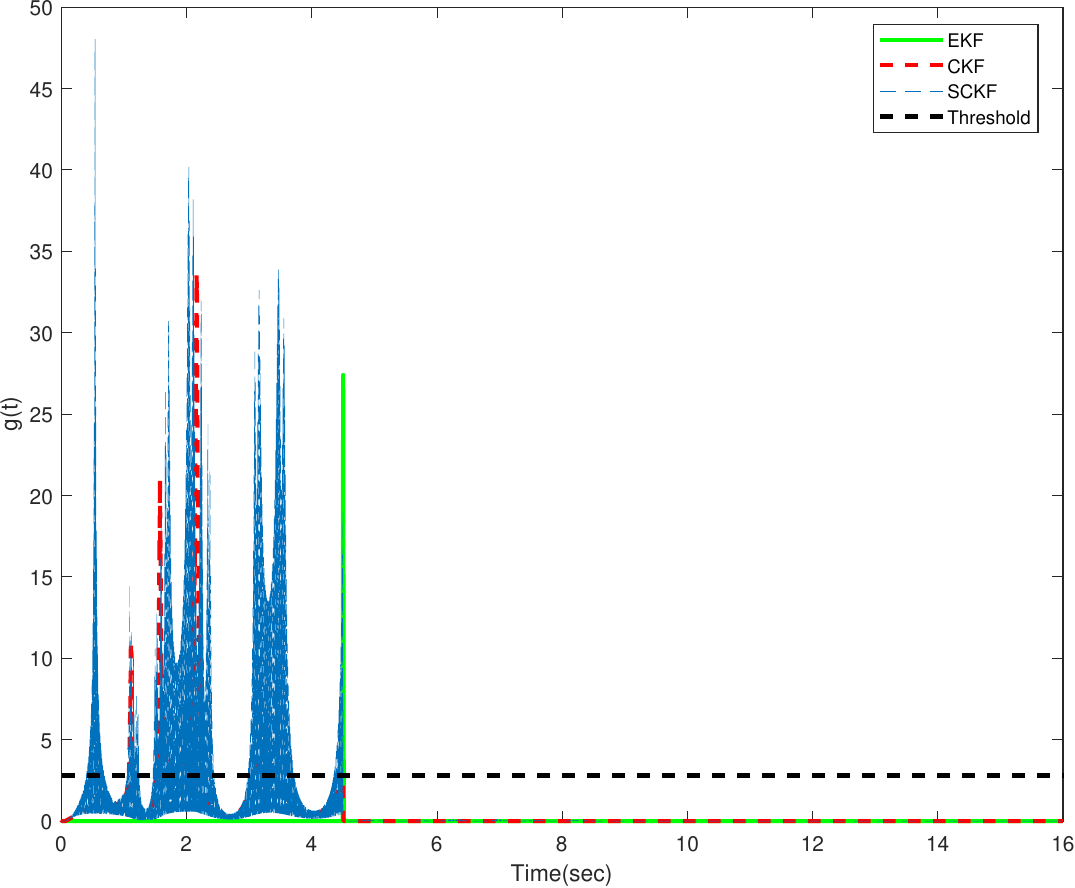}
  \end{center}
\caption{Cyber attack detection under DoS attack.}
 \label{fig:figure 9}
\end{figure} 

\begin{figure}[h]
  \begin{center}
\includegraphics[scale=0.58]{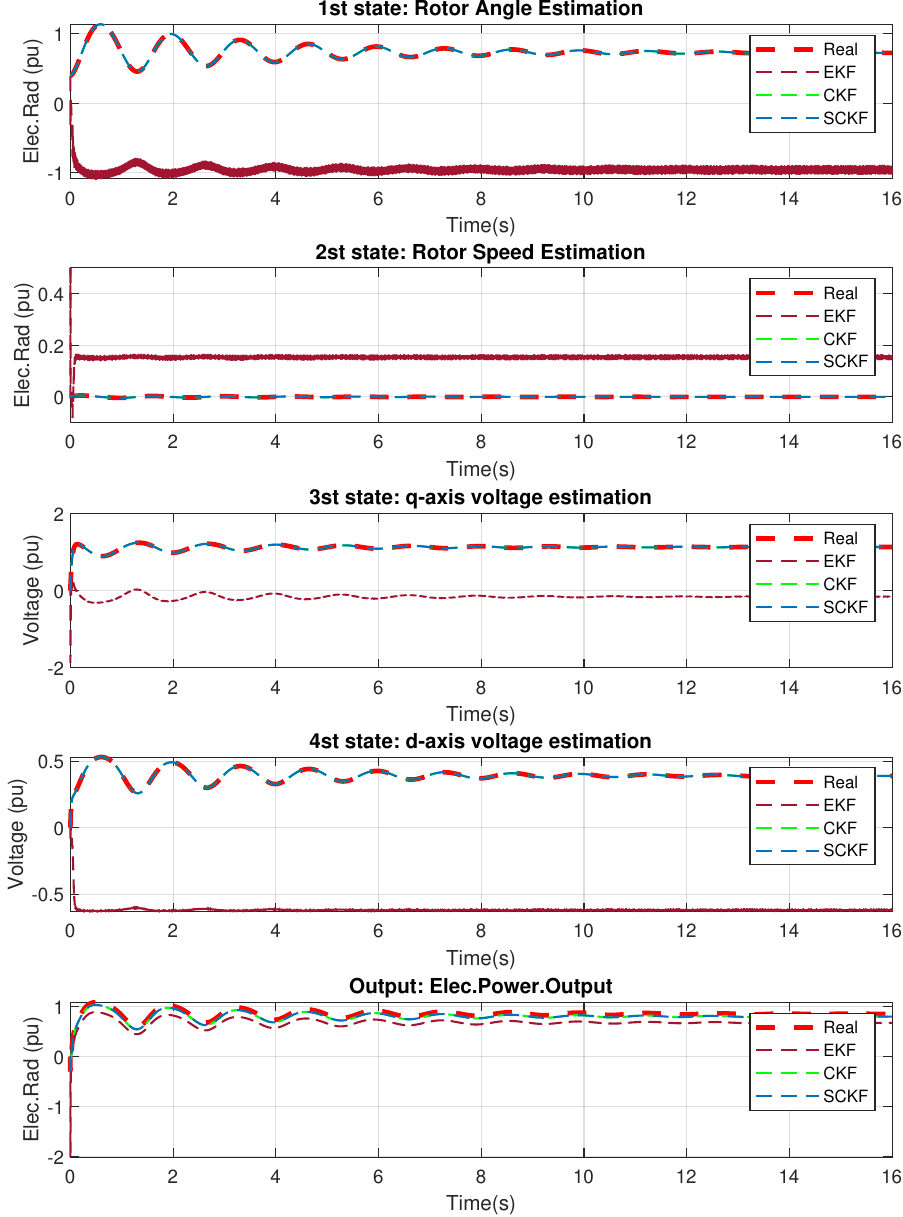}
  \end{center}
\caption{State estimation results under FDI attack}
 \label{fig:figure 10}
\end{figure} 

\begin{figure}[h]
  \begin{center}
\includegraphics[scale=0.46]{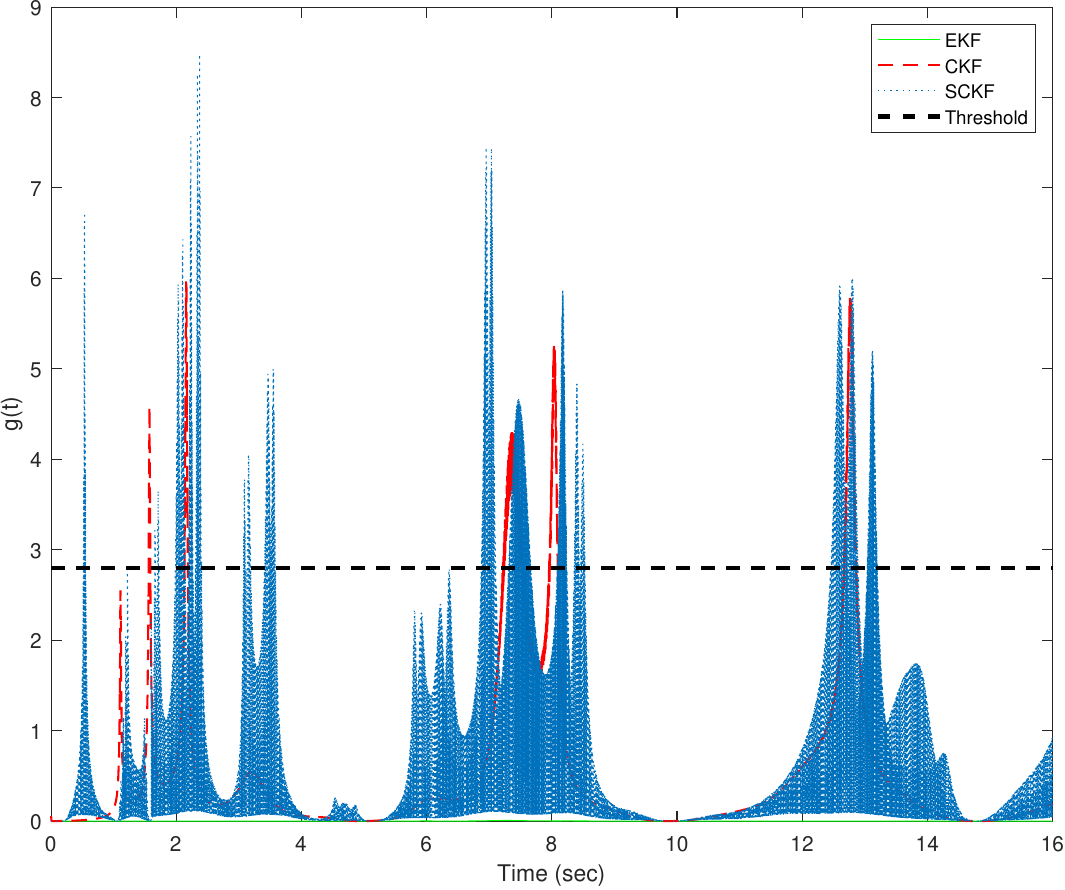}
  \end{center}
\caption{Cyber attack detection under random attack}
 \label{fig:figure 11}
\end{figure} 
\begin{figure}[h]
  \begin{center}
\includegraphics[scale=0.47]{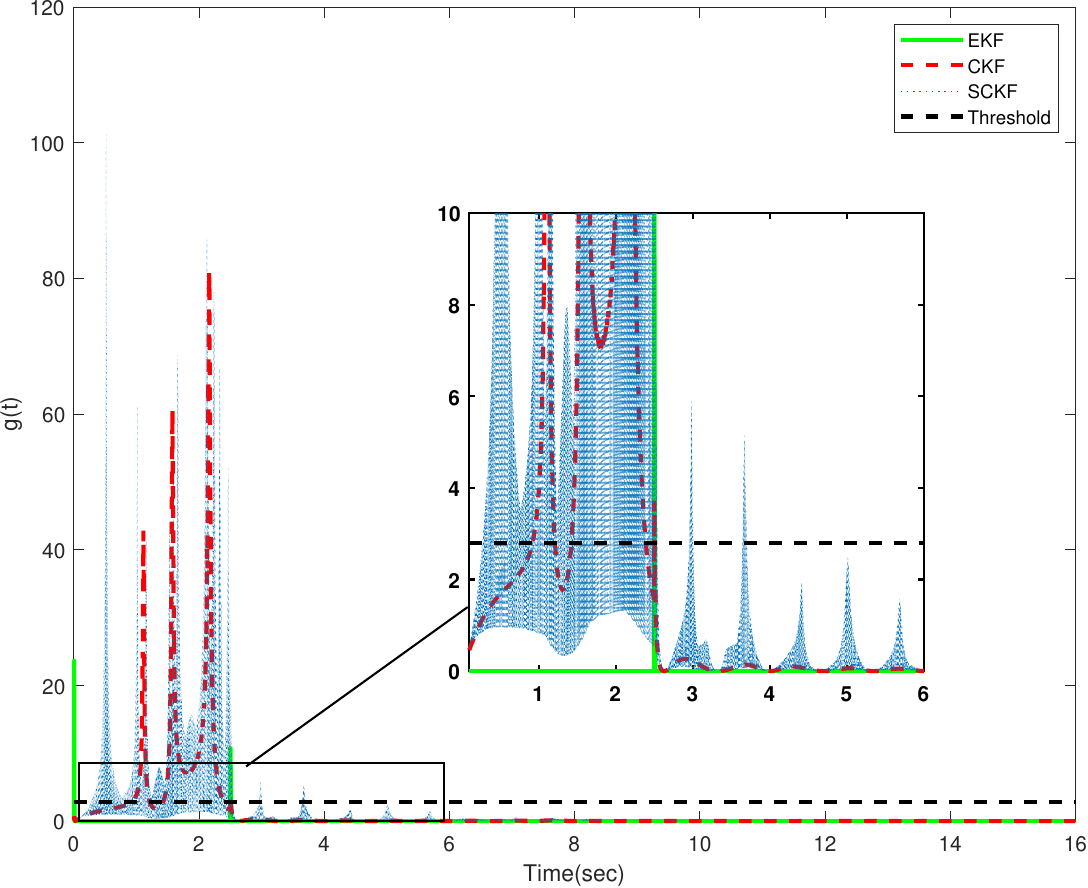}
  \end{center}
\caption{Cyber attack detection under replay attack.}
 \label{fig:figure 12}
\end{figure} 
\begin{figure}[h]
  \begin{center}
\includegraphics[scale=0.60]{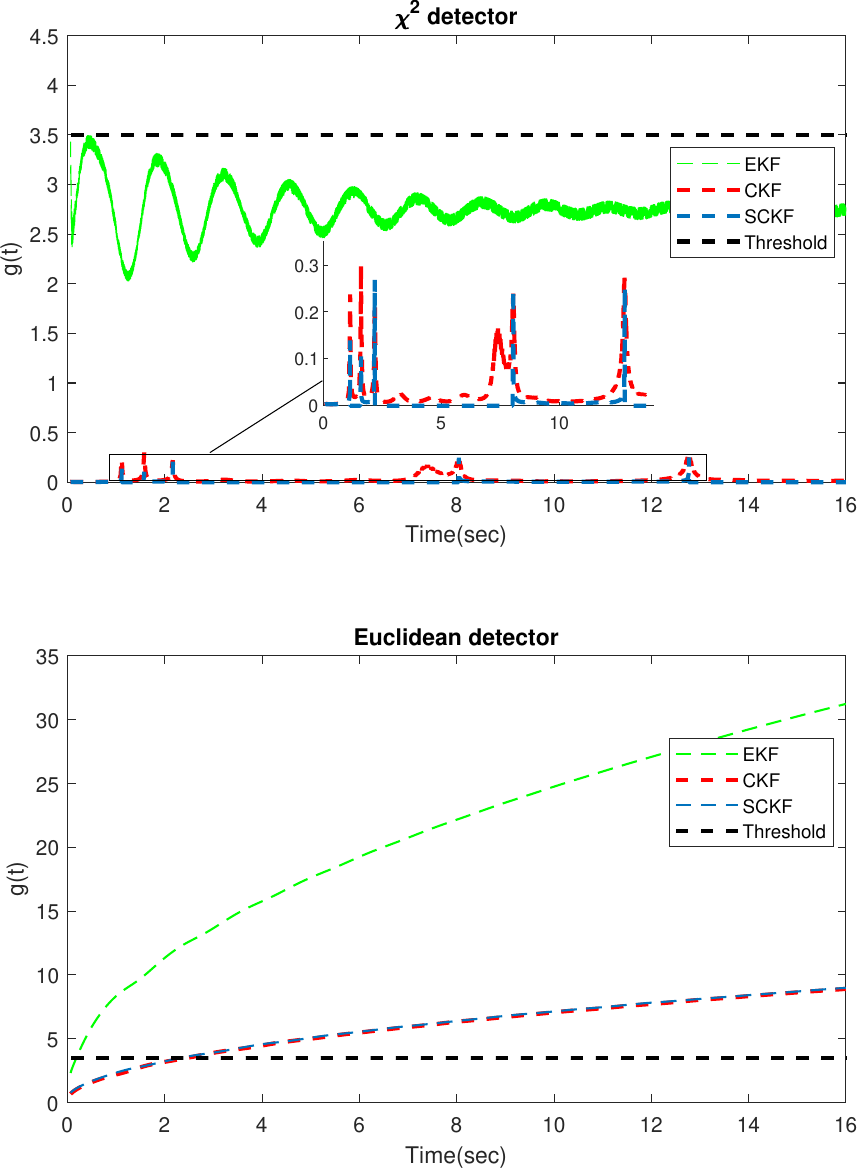}
  \end{center}
\caption{Cyber attack detection under FDI attack.}
 \label{fig:figure 13}
\end{figure}
\begin {table}[h]
\caption {Parameters and variables definition } 
\begin{center}
\begin{tabular}{llr}
\hline
Parameters   & Description & Value  \\
\hline  \vspace*{-0.75em}
\\
$D,J$  & Damping factor and inertia & 0.05,10 \\ & constant per unit  \\
$T_{m},T_{e}$          & Mechanical input and electrical torque        & 0.8     \\
$T_{do}' ,T_{qo}'$ & d-q transient open-circuit & 0.13,0.01 \\ & time constant       \\
$x_{d},x_{q}$       & d-q axis reactance     & 2.06,1.21      \\
$x_{d}',x_{q}'$       & d-q axis transient reactance     & 0.37,0.37      \\ 
$V_{t}$          & Terminal bus voltage       & No=1.02     \\
$V_{B}$          & Infinite bus voltage   & No=0.98 \\
$E_{fd}$          & Steady-state internal & No=2.29 \\ & voltage of armature       & No=2.29  \\
$\delta$   & 1st state, rotor angle       & No=0.82     \\
$\omega_{0}$  &  nominal synchronous speed      & No=377     \\
$\Delta \omega$  & 2nd state, rotor speed       & ---     \\
$e_{q}',e_{d}'$  & 3rd and 4th states, & ---\\ &  d-q transient voltage        \\
$i_{d}, i_{q}$   & d-q axis current & No=0.63, 0.51 \\
\hline
\end{tabular}
\end{center}
\end {table}

\section{Conclusion}
In this paper, we have compared the performance of two nonlinear filters, namely the EKF and the CKF under the affect of modeling error and cyber attacks. Based on theoretical analysis, we show that the CKF has better estimation accuracy than the EKF under some conditions. We also argue that the CKF is able to detect cyber attack reliably which is superior to the EKF.   
To justify our claims, we have investigated the performance of various filters on synchronous machine under different scenarios. Based on our extensive simulations, it can be seen that the CKF generates better estimation results than the EKF. In the future works, we will develop event-triggered CKF for the nonlinear dynamic system under cyber-attacks. 
\\
\\

\section*{Appendix A}
\section*{EKF ALGORITHM}

\noindent\rule{9cm}{0.4pt}
\noindent\rule{9cm}{0.4pt}
\textbf{Time  Update}
\begin{enumerate} 
 \setlength{\itemsep}{-2ex}  
 \setlength{\parskip}{0ex} 
 \setlength{\parsep}{0ex}
\item Initialization of the filter at $k=0$:\hfil\break
\begin{equation}
\left\{
                \begin{array} {l} \label{eq:equation60}
                \hat{x}_{0}^+= E(x_{0})\\
                P_{0}^+= E[(x_{0}-\hat{x}_{0}^+)(x_{0}-\hat{x}_{0}^+)^T]\\
                \end{array}
\right.
\end{equation}
where represents the expected value and the $+$ in superscript implies an \textit{a posteriori} estimate.
\\
\item Computation of the partial-derivation matrices for $k=1,2,\dots$ :
\\
\begin{equation} 
\left\{
                \begin{array} {l} \label{eq:equation61}
                F_{k-1} = \frac {\partial f_{k-1}}{\partial x}|_{\hat{x}_{k-1}^+}\\
                L_{k-1} = \frac {\partial f_{k-1}}{\partial w}|_{\hat{x}_{k-1}^+}\\
                \end{array}
\right.
\end{equation}
\\
\item Computation of the time update of the state estimate and estimation-error covariance ($k=1,2, \dots$):
\\
\begin{equation} 
\left\{
                \begin{array} {l} \label{eq:equation62}
                P_{k-1}^- =F_{k-1}P_{k-1}F_{k-1}^T+L_{k-1}Q_{k-1}L_{k-1}^T\\
                \hat{x}_{k}^- = f(\hat{x}_{k-1},u_{k-1},0)\\
                \end{array}
\right.
\end{equation}
\end{enumerate}

\noindent\rule{9cm}{0.4pt}
\textbf{Measurement Update}
\begin{enumerate}
\item Computation of the partial-derivation matrices for $k=1,2,\dots$ :
\\
\begin{equation}
\left\{
                \begin{array} {l}  \label{eq:equation63}
                H_{k-1} = \frac {\partial h_{k}}{\partial x}|_{\hat{x}_{k}^-}\\
                M_{k-1} = \frac {\partial h_{k-1}}{\partial v}|_{\hat{x}_{k}^-}\\
                \end{array}
\right.
\end{equation}
\\
\item Computation of the measurement update of the state estimate
and estimation-error covariance ($k=1,2, \dots$):
\\
\begin{equation} 
\left\{
                \begin{array} {l}  \label{eq:equation64}
                K_{k} =P_{k}^-H_{k}^T(H_{k}P_{k}^-H_{k}^T+M_{k}R_{k}M_{k}^T)^{-1}\\
                \hat{x}_{k}^+ = \hat{x}_{k}^- +K_{k}[y_{k}-h_{k}(\hat{x}_{k}^-,0)]\\
                P_{k}^+=(I-K_{k}H_{k})P_{k}^-
                \end{array}
\right.
\end{equation}
\end{enumerate}

\noindent\rule{9cm}{0.4pt}
\noindent\rule{9cm}{0.4pt}

\section*{Appendix B}
\section*{CKF ALGORITHM}

\noindent\rule{9cm}{0.4pt}
\noindent\rule{9cm}{0.4pt}
\textbf{Time Update}
\begin{enumerate}
\item Assume that the posterior density function is known at time $k$ .\\
 \begin{equation}  \label{eq:equation65}
p(x_{k-1}|D_{k-1})=N(\hat{x}_{k-1|k-1},P_{k-1|k-1})
\end{equation}
  Factorize
  \begin{equation}  \label{eq:equation66}
P_{k-1|k-1}=S_{k-1|k-1}S_{k-1|k-1}^T.
\end{equation}
  \item Compute the cubature points ($i=1,2,\dots,m$):
    \begin{equation}  \label{eq:equation67}
X_{i,k-1|k-1}=S_{k-1|k-1}\xi_{i}+\hat{x}_{k-1|k-1}
\end{equation}
where $m=2n_{x}$.
\item Compute the propagated cubature points ($i=1,2,\dots,m$)
    \begin{equation} \label{eq:equation68}
X_{i,k|k-1}^*=f(X_{i,k-1|k-1},u_{k-1}).
\end{equation}
\item Estimate the predicted state 
\begin{equation} \label{eq:equation69}
\hat{x}_{k|k-1}=\frac{1}{m}\sum_{i=1}^{m} X_{i,k|k-1}^*.
\end{equation} 
\item Estimate the predicted error covariance
    \begin{equation}  \label{eq:equation70}
P_{k|k-1}=\frac{1}{m}\sum_{i=1}^{m} X_{i,k|k-1}^*X_{i,k|k-1}^{*T}-\hat{x}_{k|k-1}\hat{x}_{k|k-1}^T+Q_{k-1}.
\end{equation}

\end{enumerate}
\noindent\rule{9cm}{0.4pt}
\textbf{Measurement Update}
\begin{enumerate}
\item Factorize
    \begin{equation}  \label{eq:equation71}
P_{k|k-1}=S_{k|k-1}S_{k|k-1}^T.
\end{equation}
\item Compute the cubature points ($i=1,2,\dots,m$)
    \begin{equation}  \label{eq:equation72}
X_{i,k|k-1}=S_{k|k-1}\xi_{i}+\hat{x}_{k|k-1}
\end{equation}
\item Compute the propagated cubature points ($i=1,2,\dots,m$)
    \begin{equation} \label{eq:equation73}
Z_{i,k|k-1}=h(X_{i,k|k-1},u_{k}).
\end{equation}
\item Estimate the predicted measurement
\begin{equation}  \label{eq:equation74}
\hat{z}_{k|k-1}=\frac{1}{m}\sum_{i=1}^{m} Z_{i,k|k-1}.
\end{equation}
\item Estimate the innovation covariance matrix 
    \begin{equation}  \label{eq:equation75}
P_{zz,k|k-1}=\frac{1}{m}\sum_{i=1}^{m} Z_{i,k|k-1}Z_{i,k|k-1}^{T}-\hat{z}_{k|k-1}\hat{z}_{k|k-1}^T+R_{k}.
\end{equation}
\item Estimate the cross-covariance matrix 
    \begin{equation}  \label{eq:equation76}
P_{xz,k|k-1}=\frac{1}{m}\sum_{i=1}^{m} X_{i,k|k-1}Z_{i,k|k-1}^{T}-\hat{x}_{k|k-1}\hat{z}_{k|k-1}^T.
\end{equation}
\item Compute the Kalman gain
    \begin{equation}  \label{eq:equation77}
W_{k}=P_{xz,k|k-1}P_{zz,k|k-1}^{-1}.
\end{equation}
\item Compute the updated state 
    \begin{equation}   \label{eq:equation78}
\hat{x}_{k|k}=\hat{x}_{k|k-1}+W_{k}(z_{k}-\hat{z}_{k|k-1}).
\end{equation}
\item Compute the corresponding error covariance 
    \begin{equation}  \label{eq:equation79}
P_{k|k}=P_{k|k-1}-W_{k}P_{zz,k|k-1}W_{k}^T.
\end{equation}
\end{enumerate}

\noindent\rule{9cm}{0.4pt}
\noindent\rule{9cm}{0.4pt}

\section*{Appendix C}
\section*{SCKF ALGORITHM}

\noindent\rule{9cm}{0.4pt}
\noindent\rule{9cm}{0.4pt}
\textbf{Time Update}
\begin{enumerate}
\item Skip the factorization \eqref{eq:equation66} because the square-root of the
error covariance, $S_{k-1|k-1}$ is available. Compute from
equation \eqref{eq:equation67} to equation \eqref{eq:equation69}.
\item Estimate the square-root factor of the predicted error
covariance
    \begin{equation} \label{eq:equation80}
S_{k|k-1}=\textbf{Trial}([\mathcal{X}_{k|k-1} S_{Q,k-1} ])
\end{equation} 
Where $S_{Q,k-1}$ is a square-root factor of $Q_{k-1}$ such
that $Q_{k-1}=S_{Q,k-1}S_{Q,k-1}^T$ and the weighted, centered
(prior mean is subtracted off) matrix
    \begin{equation}  \label{eq:equation81}
  \begin{array}{l}  
\mathcal{X}_{k|k-1}^*=\frac{1}{\sqrt{m}}[X_{1,k|k-1}^*-\hat{x}_{k|k-1} \  \   X_{2,k|k-1}^*-\hat{x}_{k|k-1} ...\\ \  \ \  \ \ \ \ \ \ \ \ \ \ \ \ \ \ \  \  \ \ \ \  \  \  \   \  \ X_{m,k|k-1}^*-\hat{x}_{k|k-1}]
  \end{array}
  \end{equation}

\end{enumerate}

\noindent\rule{9cm}{0.4pt}
\textbf{Measurement Update}
\begin{enumerate}
\item Skip the factorization \eqref{eq:equation71} because the square-root of
the error covariance, $S_{k|k-1}$, is available. Compute from
equation \eqref{eq:equation72} to equation \eqref{eq:equation74}.
\item Compute the square-root of the innovation covariance
matrix 
    \begin{equation}  \label{eq:equation82}
S_{zz,k|k-1}=\textbf{Trial}([\mathcal{Z}_{k|k-1} S_{R,k} ])
\end{equation}
where $S_{R,k}$ is a square-root factor of $R_{k}$ such that $R_{k}=S_{R,k}S_{R,k}^T$ and the weighted, centered matrix
    \begin{equation}   \label{eq:equation83}
  \begin{array}{l}
\mathcal{Z}_{k|k-1}=\frac{1}{\sqrt{m}}[Z_{1,k|k-1}-\hat{z}_{k|k-1} \  \   Z_{2,k|k-1}-\hat{z}_{k|k-1} ...\\ \  \ \  \ \ \ \ \ \ \ \ \ \ \ \ \ \ \  \  \ \ \ \  \  \  \   \  \ Z_{m,k|k-1}-\hat{z}_{k|k-1}]
  \end{array}
  \end{equation}
\item Compute the cross-covariance matrix
    \begin{equation}   \label{eq:equation84}
P_{xz,k|k-1}=\mathcal{X}_{k|k-1}\mathcal{Z}_{k|k-1}^T
\end{equation}
where the weighted, centered matrix
    \begin{equation} 
  \begin{array}{l}  \label{eq:equation85}
\mathcal{X}_{k|k-1}=\frac{1}{\sqrt{m}}[X_{1,k|k-1}-\hat{x}_{k|k-1} \  \   X_{2,k|k-1}-\hat{x}_{k|k-1} ...\\ \  \ \  \ \ \ \ \ \ \ \ \ \ \ \ \ \ \  \  \ \ \ \  \  \  \   \  \ X_{m,k|k-1}-\hat{x}_{k|k-1}]
  \end{array}
  \end{equation}
\item Compute the Kalman gain 
    \begin{equation}  \label{eq:equation86}
W_{k}=(P_{xz,k|k-1}/S_{zz,k|k-1}^T)/S_{zz,k|k-1}.
\end{equation}
\item Compute the updated state $\hat{x}_{k|k}$ as in \eqref{eq:equation78}.
\item Compute the square-root factor of the corresponding error
covariance 
    \begin{equation}  \label{eq:equation87}
S_{k|k}=\textbf{Trial}([\mathcal{X}_{k|k-1}-W_{k}\mathcal{Z}_{k|k-1}  W_{k}S_{R,k}]).
\end{equation}
\end{enumerate}

\smallskip


\begin{thebibliography}{1}


\bibitem {1}
Poovendran, Radha, et al. "Special issue on cyber-physical systems [scanning the issue]." Proceedings of the IEEE 100.1 (2012): 6-12.

\bibitem {2}
Ahmed, Syed Hassan, Gwanghyeon Kim, and Dongkyun Kim. "Cyber Physical System: Architecture, applications and research challenges." 2013 IFIP Wireless Days (WD). IEEE, 2013.

\bibitem {3}
Kim, Kyoung-Dae, and Panganamala R. Kumar. "Cyber-hysical systems: A perspective at the centennial." Proceedings of the IEEE 100.Special Centennial Issue (2012): 1287-1308.
\bibitem {4}
Lee, Edward A. "Cyber physical systems: Design challenges." 2008 11th IEEE International Symposium on Object and Component-Oriented Real-Time Distributed Computing (ISORC). IEEE, 2008.

\bibitem {5}
Sandberg, Henrik, Saurabh Amin, and Karl Henrik Johansson. "Cyberphysical security in networked control systems: An introduction to the issue." IEEE Control Systems Magazine 35.1 (2015): 20-23.

\bibitem {6}
Mo, Yilin, et al. "Cyber–physical security of a smart grid infrastructure." Proceedings of the IEEE 100.1 (2011): 195-209.

\bibitem {7}
Cárdenas, Alvaro A., Saurabh Amin, and Shankar Sastry. "Research Challenges for the Security of Control Systems." HotSec. 2008.

\bibitem{50}
T. Kargar Tasooji and H. J. Marquez, "Cooperative Localization in Mobile Robots Using Event-Triggered Mechanism: Theory and Experiments," in IEEE Transactions on Automation Science and Engineering, vol. 19, no. 4, pp. 3246-3258, Oct. 2022, doi: 10.1109/TASE.2021.3115770.


\bibitem{51}
T. K. Tasooji and H. J. Marquez, "Event-Triggered Consensus Control for Multirobot Systems With Cooperative Localization," in IEEE Transactions on Industrial Electronics, vol. 70, no. 6, pp. 5982-5993, June 2023, doi: 10.1109/TIE.2022.3192673.

\bibitem{52}
T. K. Tasooji, S. Khodadadi and H. J. Marquez, "Event-Based Secure Consensus Control for Multirobot Systems With Cooperative Localization Against DoS Attacks," in IEEE/ASME Transactions on Mechatronics, vol. 29, no. 1, pp. 715-729, Feb. 2024, doi: 10.1109/TMECH.2023.3270819.

\bibitem{53}
T. K. Tasooji and H. J. Marquez, "Decentralized Event-Triggered Cooperative Localization in Multirobot Systems Under Random Delays: With/Without Timestamps Mechanism," in IEEE/ASME Transactions on Mechatronics, vol. 28, no. 1, pp. 555-567, Feb. 2023, doi: 10.1109/TMECH.2022.3203439.

\bibitem{54}
T. Kargar Tasooji and H. J. Marquez, "A Secure Decentralized Event-Triggered Cooperative Localization in Multi-Robot Systems Under Cyber Attack," in IEEE Access, vol. 10, pp. 128101-128121, 2022, doi: 10.1109/ACCESS.2022.3227076.

\bibitem{55}
S. Khodadadi, T. K. Tasooji and H. J. Marquez, "Observer-Based Secure Control for Vehicular Platooning Under DoS Attacks," in IEEE Access, vol. 11, pp. 20542-20552, 2023, doi: 10.1109/ACCESS.2023.3250398.

\bibitem{56}
M. A. Gozukucuk et al., "Design and Simulation of an Optimal Energy Management Strategy for Plug-In Electric Vehicles," 2018 6th International Conference on Control Engineering \& Information Technology (CEIT), Istanbul, Turkey, 2018, pp. 1-6, doi: 10.1109/CEIT.2018.8751923.

\bibitem{57}
A. Mostafazadeh, T. K. Tasooji, M. Sahin and O. Usta, "Voltage control of PV-FC-battery-wind turbine for stand-alone hybrid system based on fuzzy logic controller," 2017 10th International Conference on Electrical and Electronics Engineering (ELECO), Bursa, Turkey, 2017, pp. 170-174.

\bibitem{58}
T. K. Tasooji, A. Mostafazadeh and O. Usta, "Model predictive controller as a robust algorithm for maximum power point tracking," 2017 10th International Conference on Electrical and Electronics Engineering (ELECO), Bursa, Turkey, 2017, pp. 175-179.

\bibitem{59}
T. K. Tasooji, O. Bebek, B. Ugurlu, ”A Robust Torque Controller for
Series Elastic Actuators: Model Predictive Control with a Disturbance
Observer” Turkish National Conference on Automatic Control (TOK),
Istanbul, Turkey pp. 398-402, 2017

\bibitem{60}
T. K. Tasooji, ”Energy consumption modeling and optimization of speed profile for plug-in electric vehicles”, M.Sc. dissertation, Ozyegin Univ, Istanbul, Turkey, 2018  

\bibitem{61}
T. K. Tasooji, ”Cooperative Localization and Control In Multi-Robot Systems With Event-Triggered Mechanism: Theory and Experiments”, Ph.D. dissertation, Univ. Alberta, Edmonton, AB, Canada, 2023

\bibitem{62}
S. Khodadadi, "Observer-Based Secure Control of Vehicular Platooning Under DoS attacks", M.Sc. dissertation, Univ. Alberta, Edmonton, AB, Canada, 2022

\bibitem{63}
T. K. Tasooji, S. Khodadadi, G. Liu, R. Wang, "Cooperative Control of Multi-Quadrotors for Transporting Cable-Suspended Payloads: Obstacle-Aware Planning and Event-Based Nonlinear Model Predictive Control", arXiv:2503.19135v1 [cs.RO]


\bibitem {8}
Teixeira, André, et al. "Attack models and scenarios for networked control systems." Proceedings of the 1st international conference on High Confidence Networked Systems. ACM, 2012.
\bibitem {9}
Cardenas, Alvaro A., Saurabh Amin, and Shankar Sastry. "Secure control: Towards survivable cyber-physical systems." 2008 The 28th International Conference on Distributed Computing Systems Workshops. IEEE, 2008.

\bibitem {10}
Mo, Yilin, Rohan Chabukswar, and Bruno Sinopoli. "Detecting integrity attacks on SCADA systems." IEEE Transactions on Control Systems Technology 22.4 (2013): 1396-1407.

\bibitem {11}
Huang, Zhenyu, Kevin Schneider, and Jarek Nieplocha. "Feasibility studies of applying Kalman filter techniques to power system dynamic state estimation." 2007 International Power Engineering Conference (IPEC 2007). IEEE, 2007.

\bibitem {12}
Ghahremani, Esmaeil, and Innocent Kamwa. "Dynamic state estimation in power system by applying the extended Kalman filter with unknown inputs to phasor measurements." IEEE Transactions on Power Systems 26.4 (2011): 2556-2566.

\bibitem {13}
Valverde, Gustavo, and Vladimir Terzija. "Unscented Kalman filter for power system dynamic state estimation." IET generation, transmission \& distribution 5.1 (2011): 29-37.

\bibitem {14}
Ghahremani, Esmaeil, and Innocent Kamwa. "Online state estimation of a synchronous generator using unscented Kalman filter from phasor measurements units." IEEE Transactions on Energy Conversion 26.4 (2011): 1099-1108.

\bibitem {15}
Wang, Shaobu, Wenzhong Gao, and AP Sakis Meliopoulos. "An alternative method for power system dynamic state estimation based on unscented transform." IEEE transactions on power systems 27.2 (2011): 942-950.

\bibitem {16}
Singh, Abhinav Kumar, and Bikash C. Pal. "Decentralized dynamic state estimation in power systems using unscented transformation." IEEE Transactions on Power Systems 29.2 (2013): 794-804.

\bibitem {17}
Sun, Kai, Junjian Qi, and Wei Kang. "Power system observability and dynamic state estimation for stability monitoring using synchrophasor measurements." Control Engineering Practice 53 (2016): 160-172.

\bibitem {18}
Van Der Merwe, Rudolph, and Eric A. Wan. "The square-root unscented Kalman filter for state and parameter-estimation." 2001 IEEE international conference on acoustics, speech, and signal processing. Proceedings (Cat. No. 01CH37221). Vol. 6. IEEE, 2001.

\bibitem {19}
Qi, Junjian, Kai Sun, and Wei Kang. "Optimal PMU placement for power system dynamic state estimation by using empirical observability Gramian." IEEE Transactions on power Systems 30.4 (2014): 2041-2054.

\bibitem {20}
Qi, Junjian, et al. "Dynamic state estimation for multi-machine power system by unscented Kalman filter with enhanced numerical stability." IEEE Transactions on Smart Grid 9.2 (2016): 1184-1196.

\bibitem {21}
Qi, Junjian, Kai Sun, and Wei Kang. "Adaptive optimal PMU placement based on empirical observability gramian." IFAC-PapersOnLine 49.18 (2016): 482-487.

\bibitem {22}
Zhou, Ning, Da Meng, and Shuai Lu. "Estimation of the dynamic states of synchronous machines using an extended particle filter." IEEE Transactions on Power Systems 28.4 (2013): 4152-4161.

\bibitem {23}
Cui, Yinan, and Rajesh Kavasseri. "A particle filter for dynamic state estimation in multi-machine systems with detailed models." IEEE Transactions on Power Systems 30.6 (2015): 3377-3385.

\bibitem {24}
Zhou, Ning, et al. "Dynamic state estimation of a synchronous machine using PMU data: A comparative study." IEEE Transactions on Smart Grid 6.1 (2014): 450-460.

\bibitem {25}
Guo, Ziyang, et al. "Optimal linear cyber-attack on remote state estimation." IEEE Transactions on Control of Network Systems 4.1 (2017): 4-13.

\bibitem {26}
Shi, Dawei, Robert J. Elliott, and Tongwen Chen. "On finite-state stochastic modeling and secure estimation of cyber-physical systems." IEEE Transactions on Automatic Control 62.1 (2017): 65-80.

\bibitem {27}
Xie, Le, Yilin Mo, and Bruno Sinopoli. "Integrity data attacks in power market operations." IEEE Transactions on Smart Grid 2.4 (2011): 659-666.

\bibitem {28}
Yang, Qiang, et al. "PMU placement in electric transmission networks for reliable state estimation against false data injection attacks." IEEE Internet of Things Journal 4.6 (2017): 1978-1986.

\bibitem {29}
Gandhi, Mital A., and Lamine Mili. "Robust Kalman filter based on a generalized maximum-likelihood-type estimator." IEEE Transactions on Signal Processing 58.5 (2009): 2509-2520.

\bibitem {30}
Taha, Ahmad F., et al. "Risk mitigation for dynamic state estimation against cyber attacks and unknown inputs." IEEE Transactions on Smart Grid 9.2 (2016): 886-899.

\bibitem {31}
Karimipour, Hadis, and Venkata Dinavahi. "Robust massively parallel dynamic state estimation of power systems against cyber-attack." IEEE Access 6 (2017): 2984-2995.


\bibitem {32}
Zhao, Junbo, Lamine Mili, and Antonio Gomez-Exposito. "Constrained Robust Unscented Kalman Filter for Generalized Dynamic State Estimation." IEEE Transactions on Power Systems (2019).


\bibitem {33}
Manandhar, Kebina, et al. "Detection of faults and attacks including false data injection attack in smart grid using Kalman filter." IEEE transactions on control of network systems 1.4 (2014): 370-379.

\bibitem {34}
Qi, Junjian, Ahmad F. Taha, and Jianhui Wang. "Comparing Kalman filters and observers for power system dynamic state estimation with model uncertainty and malicious cyber attacks." IEEE Access 6 (2018): 77155-77168.

\bibitem {35}
Ghahremani, Esmaeil, and Innocent Kamwa. "Dynamic state estimation in power system by applying the extended Kalman filter with unknown inputs to phasor measurements." IEEE Transactions on Power Systems 26.4 (2011): 2556-2566.

\bibitem {36}
Ghahremani, Esmaeil, and Innocent Kamwa. "Online state estimation of a synchronous generator using unscented Kalman filter from phasor measurements units." IEEE Transactions on Energy Conversion 26.4 (2011): 1099-1108.

\bibitem {37}
Zhao, Yingwei. "Performance evaluation of cubature Kalman filter in a GPS/IMU tightly-coupled navigation system." Signal Processing 119 (2016): 67-79.

\bibitem {38}
Arasaratnam, Ienkaran, and Simon Haykin. "Cubature kalman filters." IEEE Transactions on automatic control 54.6 (2009): 1254-1269.

\bibitem {39}
R. E. Bellman, Adaptive Control Processes. Princeton, NJ: Princeton
Univ. Press, 1961. 

\bibitem {40}
Simon, Dan. Optimal state estimation: Kalman, H infinity, and nonlinear approaches. John Wiley \& Sons, 2006.

\bibitem {41}
Zhou, Kemin, John Comstock Doyle, and Keith Glover. Robust and optimal control. Vol. 40. New Jersey: Prentice hall, 1996.

\bibitem {42}
Rapp, Knut, and Per-Ole Nyman. "Stability properties of the discrete-time extended Kalman filter." IFAC Proceedings Volumes 37.13 (2004): 1377-1382

\bibitem {43}
Lee, A. "Electric sector failure scenarios and impact analyses." National Electric Sector Cybersecurity Organization Resource (NESCOR) Technical Working Group 1 (2013).

\bibitem {44}
Sridhar, Siddharth, Adam Hahn, and Manimaran Govindarasu. "Cyber-physical system security for the electric power grid." Proceedings of the IEEE 100.1 (2011): 210-224.

\bibitem {45}
Mo, Yilin, and Bruno Sinopoli. "False data injection attacks in control systems." Preprints of the 1st workshop on Secure Control Systems. 2010.

\bibitem {46}
Ding, Derui, et al. "Security control for discrete-time stochastic nonlinear systems subject to deception attacks." IEEE Transactions on Systems, Man, and Cybernetics: Systems 48.5 (2016): 779-789.

\bibitem {47}
Brumback, B., and M. Srinath. "A chi-square test for fault-detection in Kalman filters." IEEE Transactions on Automatic Control 32.6 (1987): 552-554.

\bibitem {48}
Reif, Konrad, et al. "Stochastic stability of the discrete-time extended Kalman filter." IEEE Transactions on Automatic control 44.4 (1999): 714-728.

\bibitem {49}
Hu, Jun, et al. "Extended Kalman filtering with stochastic nonlinearities and multiple missing measurements." Automatica 48.9 (2012): 2007-2015.


\end{thebibliography}
\end{document}